\documentclass[final]{article}

\usepackage{amssymb} \usepackage{amsmath} \usepackage{amsthm} \usepackage{algorithm} \usepackage{algorithmic} \usepackage{graphicx} \usepackage{subfigure} \usepackage{url} \usepackage[margin=1in]{geometry}

\newtheorem{theorem}{Theorem} \newtheorem{definition}{Definition}   \newtheorem{corollary}{Corollary}

\begin{document}

\title{On Time-Sensitive Revenue Management and Energy Scheduling in Green Data Centers}

\author{Huangxin Wang\thanks{Department of Computer Science, George Mason University. Email: \textsf{hwang14@gmu.edu}} \and Jean X. Zhang\thanks{School of Business, Virginia Commonwealth University. Email: \textsf{jxzhang@vcu.edu}} \and Fei Li\thanks{Department of Computer Science, George Mason University. Email: \textsf{lifei@cs.gmu.edu}}}

\date{}

\maketitle


\begin{abstract}
In this paper, we design an analytically and experimentally better online energy and job scheduling algorithm with the objective of maximizing net profit for a service provider in green data centers. We first study the previously known algorithms and conclude that these online algorithms have provable poor performance against their worst-case scenarios. To guarantee an online algorithm's  performance in hindsight, we design a randomized algorithm to schedule energy and jobs in the data centers and prove the algorithm's expected competitive ratio in various settings. Our algorithm is theoretical-sound and it outperforms the previously known algorithms in many settings using both real traces and simulated data. An optimal offline algorithm is also implemented as an empirical benchmark.
\end{abstract}


\section{Introduction}

A \emph{data center} is a computing facility used to house computer systems and associated components such as communication and storage subsystems. Usually, a data center stores data and provides computing facilities to its customers. Through charging fees for data accessing and computing services, a data center gains revenue~\cite{amazonprice}. At the same time, to maintain its running structure, a data center has to pay \emph{operational costs}, including hardware costs (such as of upgrading computing and storage devices and air conditioning facilities), electrical bills for power supply, network connection costs, in addition to personnel costs. To maximize a data center's net profit, we expect to increase the revenue collected and decrease the operational cost paid concurrently.

The ever increasing power costs and energy consumption in data centers have brought with many serious economic and environmental problems to our society and evoked significant attention recently. As reported, the energy consumption of all data centers consisted of 10\% of the total U.S. energy consumption in 2006 and has increased 56\% over the past five-year period~\cite{epa}. The estimates of annual power costs for U.S. data centers in 2010 reached as high as $3.3$ billion dollars~\cite{EnergyCost}. As an example, in a modern high-scale data center with 45,000 to 50,000 servers, more than 70\% of its operational cost (around half a billion dollars per year)~\cite{usage} goes to maintaining the servers and providing power supply. Targeting on both economic and environmental factors, academic researchers and industrial policy makers have investigated revenue management policies and engineering solutions to make data centers work better without sacrificing service qualities and environment sustainability.

A growing trend of reducing energy costs as well as protecting environments is to fuel a data center using renewable energy such as wind and solar power. We term this type of energy as ``green energy'' as it comes from renewable and non-polluting sources. Unfortunately, the amount of green energy is usually intermittent, limited and cannot be fully predicted in the long term. Another type of energy, called ``brown energy'', comes from the available electrical grid in which the power is produced by carbon-intensive means. We would like to minimize the usage of brown energy, although its supply is usually regarded unlimited. A data center with both green energy and brown energy supplies is called a \emph{green data center}. Due to economic concerns and technical difficulties, no battery is assumed to be available to store any surplus green energy~\cite{Bianchini12}.

In this paper, we consider a job and energy scheduling problem in green data centers. The ultimate goal is to optimize green/brown energy usage without sacrificing service qualities. Our work is built upon the study by Goiri~\emph{et. al}~\cite{GoiriL11}. In this problem, jobs arrive at a data center over time. We design \emph{revenue management} algorithms whose task is to determine \emph{whether}, \emph{when} and \emph{where} to schedule a job request from customers. Committing and finishing a job earns the service provider some revenue. Note that in completing a job, different ways of designating machines, types of energy, and time intervals may result in different operational costs. We target on the question: \emph{How to wisely dispatch jobs and schedule energy to maximize the net profit achieved by a data center's service provider?} Recall that the information on later released jobs and future generated green energy is in general unknown beforehand, what we study in this paper can be regarded as an online version of a multiple machine scheduling problem.

To evaluate an online scheduling algorithm's performance, we address two metrics from two perspectives. In theory, we use \emph{competitive ratio}~\cite{BorodinE98} to measure an online algorithm's worst-case performance against an adversarial clairvoyant. Competitive analysis has been used widely to analyze online algorithms in computer science and operations research. In practice, we conduct experiments using both real traces and simulated data. The crux of our algorithm's idea is to introduce `randomness' in scheduling energy and jobs. As what we will see in the remaining parts of this paper, `randomness' helps both theoretically and empirically, particularly in adversarial settings.


\subsection{Problem formulation}

In data centers, a service provider is regarded as a \emph{resource provider} which provides a set of machines that will be shared and used by the data centers' clients. The clients, regarded as \emph{resource consumers}, have their jobs processed and in turn, pay the service provider for the service they get. The service provider's \emph{revenue management} has the objective of maximizing its \emph{net profit}, defined as the difference between the revenue collected from the clients and the operational costs charged to maintain the computing system. Here the operational costs do not include those for upgrading systems, paying personnel, or training operators.

We model the service provider's revenue management as a job and energy scheduling problem. The components of a computing system within a data center is pictured in Figure~\ref{fig:system} and we introduce each of them in details as below.

\begin{figure}[h!]
\center
\includegraphics[width=.4\textwidth]{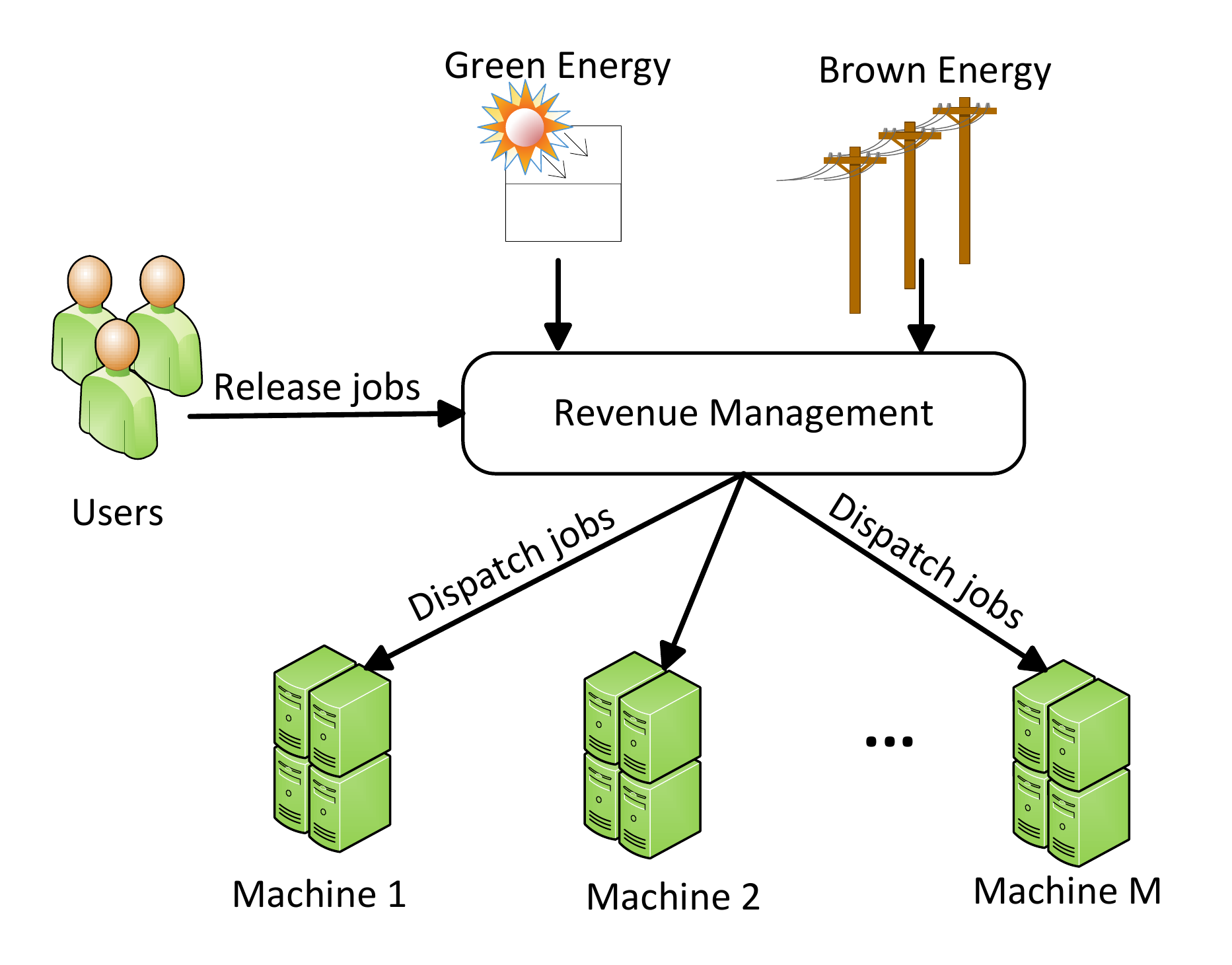}
\caption{Components of a solar-powered green data center.}
\label{fig:system}
\end{figure}

\paragraph*{Machine resources}

Time is discrete. A service provider has $M$ \emph{machines} (also called \emph{nodes}) to schedule jobs. At any time, a node can process at most one job. To make these machines function, electrical power resource is consumed at the time when jobs are being executed. Without loss of generality, we assume that a node consumes $1$ unit of energy per time slot when it is processing a job and $0$ unit otherwise.

\paragraph*{Jobs}

Clients (customers) release jobs to be processed. Jobs arrive over time in an online manner. At a time, some (may be $0$) jobs arrive. Each job $j$ has an integer \emph{arriving time} (also called \emph{release time}) $r_j \in \mathbb{Z}^+$, an integer \emph{deadline} $d_j \in \mathbb{Z}^+$, an integer \emph{processing time} $p_j \in \mathbb{Z}^+$, and an integer \emph{node requirement} $q_j \in \mathbb{Z}^+$. It takes $p_j$ time units to complete job $j$. Running one job may require more than one nodes to be simultaneously active at a time. The node requirement $q_j (\geq 1)$ indicates the number of nodes that a job $j$ needs when it is being executed. The total \emph{machine resource requirement} for a job $j$ is thus $q_j \cdot p_j$. Jobs may or may not be executed within a consecutive time interval and we call these settings \emph{job non-preemptive manner} and \emph{job preemptive manner} respectively.

\paragraph*{Deadline-driven revenue}

The clients pay to the service provider for their service received. In general, the payoff depends on the job's machine resource requirement. For each job that has been completed within the data center, the client pays for a fee proportional to the job's resource requirement. We assume that a client pays $\$c \cdot q_j \cdot p_j$ upon completion by its deadline and $\$0$ upon no completion by its deadline. Here $c$ is called a \emph{service charging rate}, for instance, as what is specified by Amazon EC2~\cite{amazonprice}.

\paragraph*{Time-sensitive costs}

Energy is consumed along the course of nodes executing jobs. There are two types of energy resources: \emph{green energy} and \emph{brown energy}. Usually, a system is able to predict green energy quantity only within a 48-hour \emph{scheduling window}. In~\cite{GoiriL11}, a scheduling window was defined as a time interval of $48$ hours, which was further divided into time slots with length of $15$ minutes. We in general assume that the brown energy supply is unlimited.

Different types of energy cost vary over time. We assume that green energy costs us price $\$0$ per machine time slot. While brown energy's unit-cost is time-sensitive and thus it is a variable related to on-peak/off-peak time periods. A unit of brown energy has price $\$B^d$ when at on-peak (usually at daytime) and price $\$B^n$ when at off-peak (usually at nighttime). This assumption is the most common one used in modeling brown electricity pricing~\cite{GoiriL11}. For example, the prices charged by an integrated generation and energy service company in New Jersey~\cite{GoiriL11} are  $\$0.13/kWh$ and $\$0.08/kWh$ at on-peak (from $9$am to $11$pm) and at off-peak (from $11$pm to $9$am) respectively.

\paragraph*{Objective}

Scheduling jobs successfully can earn the service provider some \emph{revenue} (also called \emph{job values}). However, if we pay for the brown energy used in additional to the limited green energy to power the data centers to complete jobs, we have to pay an electrical bill as our \emph{operational costs}. We define
\begin{displaymath}
\mbox{net profit = revenue - operational cost},
\end{displaymath}
where \emph{revenue} is the total job value that we gain through finishing jobs and operational cost is the total brown energy cost that the service provider consumes to run these machines. The objective of revenue management for a service provider within green data centers is to design a scheduler to complete all or part of the released jobs in order to maximize net profit. We call this problem GDC-RM, standing for `Green Data Center's Revenue Management'.

In the remaining parts of this paper, we present combinatorial optimization algorithms for GDC-RM. As in general the job arriving information is unknown beforehand, GDC-RM is essentially an online decision-making problem. For reference, notations used in this paper are summarized in Table~\ref{tab_notations}.

\begin{table}[!ht]
\centering
\begin{tabular}{|p{1.2cm}|p{6.5cm}|}
\hline
notation & meanings\\ \hline \hline
$p_j$ & job $j$'s processing time\\ \hline
$q_j$ & job $j$'s node requirements\\ \hline
$r_j$ & job $j$'s arriving time\\ \hline
$d_j$ & job $j$'s deadline\\ \hline
$M$ & number of machines/nodes\\ \hline
$B^d$ & on-peak (at daytime) power price per unit brown energy\\ \hline
$B^n$ & off-peak (at nighttime) power price per unit brown energy\\ \hline
\end{tabular}
\caption{Notations used in this paper and their meanings.}
\label{tab_notations}
\end{table}


\subsection{Related work}

People have worked on how to use green energy in green data centers in an efficient and effective manner. Although green energy has the advantages of being cost-effective and environmental-friendly, there is a challenge in using it due to their daily seasonal variability. Another challenge is due to customers' workload fluctuations~\cite{HeY10}. There could be a mismatch between the green energy supply and the workload's energy demand in the time axis --- a heavy workload arrives when the green energy supply is low. One solution is to ``bank'' green energy in batteries or on the grid itself~\cite{Bianchini12} for later possible use. However, this approach incurs huge energy lost and high additional maintenance cost~\cite{Bianchini12}. Thus, an online matching of workload and energy is demanded for green data centers.

The research on scheduling energy and jobs in an online manner has attracted a lot of attentions. Two data center settings have been considered: (1) centralized data centers~\cite{GoiriL11, GoiriL12, Krioukov11, AksanliV11, LiQ11}, and (2) geographically distributed data centers~\cite{LiuL11,ChenH12,LinL12,ZhangW11,LeB09,LeN10}. The objectives to optimize are usually classified as (a) to maximize green energy consumption~\cite{Krioukov11, GoiriL11, GoiriL12, ZhangW11, LeB09, LeN10, AksanliV11}; (b) to minimize brown energy consumption or cost~\cite{ChenH12, GoiriL11, GoiriL12, LiuC12,LinL12}; and (c) to maximize profits~\cite{GhamkhariR13}. In addition, some researchers incorporated the dynamic pricing of brown energy~\cite{GoiriL11, GoiriL12, RaoL10} in their problem models.

Among the research on centralized data centers, Goiri~\emph{et al.}~\cite{GoiriL11} proposed a greedy parallel batch job scheduler for a data center powered by solar energy with the goal of maximizing green energy power consumption. They further integrated the green scheduling in Hadoop~\cite{GoiriL12}. Krioukov~\emph{et al.}~\cite{Krioukov11} studied data intensive applications and proposed a scheduling algorithm with the goal of maximizing green energy consumption while satisfying job deadlines. Aksanli~\emph{et al.}~\cite{AksanliV11} developed a green-aware scheduling algorithm for both online service and batch jobs aiming at improving green energy usage.  Liu~\emph{et al.}~\cite{LiuC12} studied workload and cooling management with the goal to reduce brown energy costs.  The algorithms underlying these solutions are known as \emph{First-Fit} and \emph{Best-Fit}. For an arriving job, the First-Fit algorithm finds the earliest available time slots to schedule the job according to its resource requirements, while the Best-Fit algorithm locates the most `\emph{cost-efficient}' time slots to schedule the job. The First-Fit algorithm in general ignores the cost difference in scheduling jobs at various time intervals. The Best-Fit algorithm picks up the \emph{best} time interval to schedule a job in a myopic way and it does not take later job arrivals or energy supplies into account. Different from the previous study, our research in this paper is to find an ideal tradeoff between these two algorithms by introducing randomness. We prove the algorithm's theoretical bounds and also demonstrate the performance improvement.

Research on geographical data centers focuses on distributing workload among distributed data centers in order to consume the available free green energy or relative cheaper brown energy at other data centers. Chen~\emph{et al.}~\cite{ChenH12} proposed a centralized scheduler that migrates workload across geographical data centers according to the green energy supply at different data centers. Lin \emph{et al.}~\cite{LinL12} proposed online algorithms for scheduling workloads across geographical data center with the goal to minimize total energy cost. Although the proposed algorithm did reduce the energy cost but the total energy consumption increased. Liu~\emph{et al.}~\cite{LiuL11} further studied how the geographical load balancing and the proportional brown energy pricing scheme could help encourage the use of green energy and reduce the use of brown energy. Zhang~\emph{et al.}~\cite{ZhangW11} and Le~\emph{et al.}~\cite{LeB09, LeN10} researched on scheduling online services across multiple data centers to maximize green energy consumption.

Although geographical data centers have become popular nowadays for big companies such as Google, Amazon, a small centralized  data center is still important since as reported, numerous small and medium-sized companies are the main contributors to the energy consumed by data centers~\cite{epa}. On one hand, small data centers owned by small organizations usually have less efficient energy-efficient management strategies compared to those big companies. On the other hand, the sizes of small and medium data centers are of numerous amount. These data centers in small or medium-sized companies may range from a few dozen servers housed in a machine room to several hundreds of servers housed in a large enterprise installation. Therefore, there is a huge impact in studying the profit maximization problem for centralized data centers.

Most of prior work focuses on either maximizing green energy consumption or minimizing brown energy consumption/cost except~\cite{GhamkhariR13} which studied the net profit maximization problem for centralized data center service providers.  Actually there is a trade-off between the minimization of  energy expenditure and the maximization of net profit. ~\cite{GhamkhariR13} proposed a systematic approach to maximize green data center's profit with a stochastic assumption on the workload. The workload that they studied is restricted to online service requests with variable arrival rates. In this paper, we study the profit maximization problem in a more general setting. In particular, we make no particular assumptions over the workload's stochastic property and we allow the workloads to include a batch job which requests to be simultaneously executed on multiple nodes. In addition, we incorporate varying brown energy prices in our model.


\section{Algorithms}

The problem GDC-RM is computational hard as shown in the appendix, even we have all information including the future released jobs and later generated green energy beforehand.

In reality, job scheduling in data centers is essentially an online problem. For the problem GDC-RM, we first discuss two widely-used heuristic online algorithms First-Fit and Best-Fit and analyze their limitations. Then we propose a randomized algorithm Random-Fit.  We conduct \emph{competitive analysis} when we evaluate an online algorithm's theoretical performance. Competitive analysis is to compare the output of an online algorithm with that of an optimal offline clairvoyant algorithm. This unrealistic offline algorithm is assumed to know all the input information (including the solar energy arrivals, brown energy prices, and job arriving sequences) beforehand.

\begin{definition}[Competitive ratio~\cite{BorodinE98}]
A deterministic (respectively, randomized) online algorithm ON is called $k$-competitive if its (respectively, expected) performance of any instance is at least $1 / k$ times of that of an optimal offline algorithm. The optimal offline algorithm is also called (respectively, oblivious) adversary. Let OPT denote the optimal offline solution of an input. Competitive ratio $k$ is defined as
\begin{displaymath}
k := \max \limits_{I} \frac{OPT - \delta}{E[ON]},
\end{displaymath}
where $\delta$ is a constant and $E[ON]$ is ON's (expected) output of an input.
\end{definition}

Note that unlike  stochastic algorithms which heavily rely on the statistical assumptions on the input sequence, competitive online algorithms guarantee the worst-case performance in any given finite time frame.

Figure~\ref{fig:competitiveRatio}, as an example, illustrates the advantage of evaluating online algorithms using competitive analysis. In Figure~\ref{fig:competitiveRatio}, $y$-axis represents online algorithms' performance and $x$-axis shows various types of workloads. There are 3 algorithms: ALG1, ALG2, and ALG3. Compared with ALG1 and ALG3 which perform really bad for workload W1 and workload W3 respectively, ALG2 has the best competitive ratio and it keeps its performance the best against its worst-case scenarios for any arbitrary workload. Competitive analysis is used when rigorous analysis of online algorithms is needed and when the input's stochastic properties are hard to get. For the problem GDC-RM, a green data center's workloads are hard to model~\cite{MeisnerW10} and thus competitive analysis is a suitable metric.

\begin{figure}[h!]
\center
\includegraphics[width=.4\textwidth]{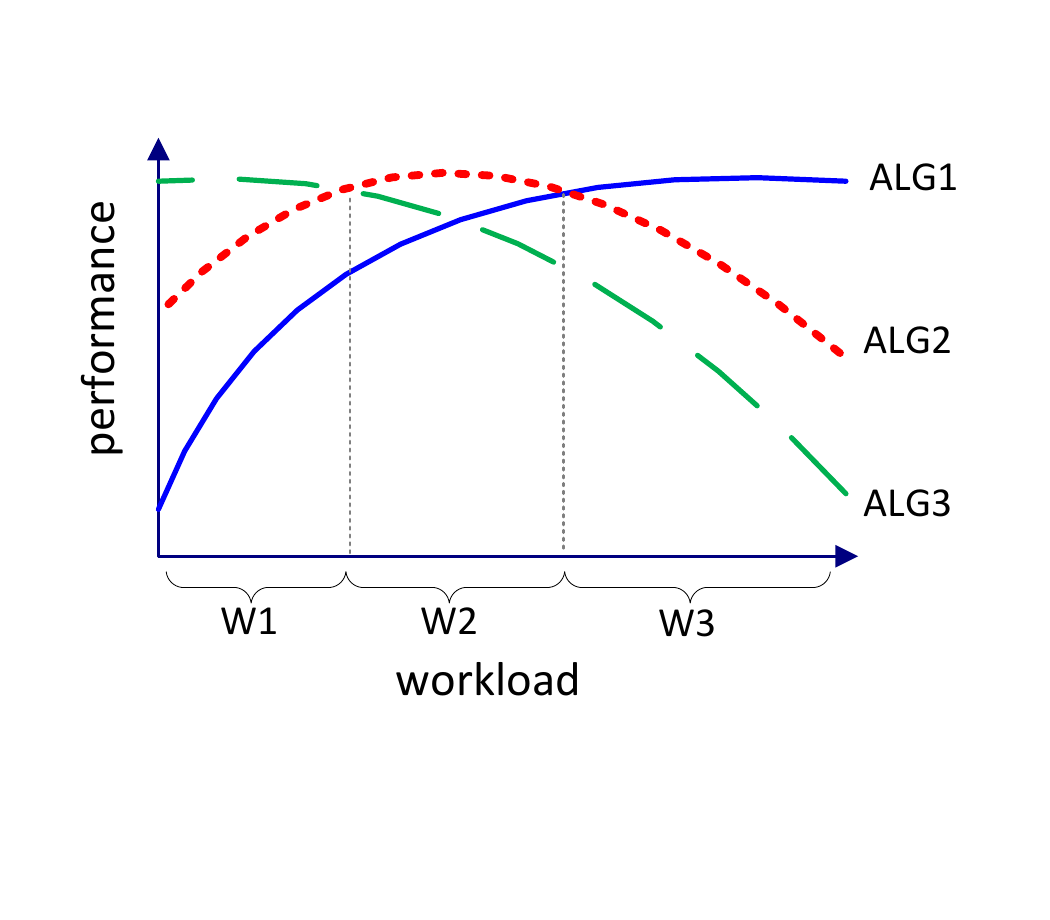}
\caption{An example illustrating the advantage of competitive analysis.}
\label{fig:competitiveRatio}
\end{figure}


\subsection{Competitive analysis of First-Fit and Best-Fit}

First-Fit is a conventional deterministic online scheduler which tries to schedule a job to the earliest available time slots regardless of its energy cost. Although this approach can cause minimum delay of a job, it might not achieve a good overall profit due to large brown energy cost needed (instead of using green energy or less-expensive night-time brown energy) to finish the job in earlier time slots.

Best-Fit is another conventional deterministic online scheduler based on the greedy idea. This algorithm always assigns jobs to the most \emph{cost-efficient} time slots. As points out in ~\cite{GoiriL11}, Best-Fit may reject more jobs or miss more deadlines than First-Fit does. The reason lies at the observation that Best-Fit always delays jobs to the best cost-efficient time slots regardless of future workload for those slots. As a result, some jobs may fail to be scheduled due to deadline constraints and thus the profit is harmed.

We include the competitive analysis of First-Fit and Best-Fit in the appendix. As shown there, First-Fit and Best-Fit have arbitrary worse competitive ratios. Even for the special case in which all jobs are with the same processing times and node requirements (in this case, all jobs have the same revenue collected), First-Fit and Best-Fit have lower bounds of competitive ratio $2$. As what we can see from the analysis of First-Fit and Best-Fit, a deterministic online algorithm is pessimistic. The crux of analysis lies as below: On one hand, if we schedule a job regardless of its energy cost, then we prefer the algorithm First-Fit as it leaves room for later arriving jobs to be scheduled. One the other hand, if we schedule a job considering its energy cost, then this job may be scheduled at a later time (e.g., the green energy runs out for now and there exists predicted green energy in the future) and thus later released jobs may not be admitted and completed on time.


\subsection{A randomized algorithm Random-Fit}

In order to solve this dilemma introduced by First-Fit and Best-Fit in maximizing net profit, we introduce an algorithm with internal randomness to twist the high brown energy cost that we pay right now (using First-Fit) and the high cost of losing potential future jobs (using Best-Fit). We develop an algorithm called Random-Fit (as in Algorithm~\ref{alg_Random-Fit}).

\begin{algorithm}
\caption{Random-Fit (RF)}
\begin{algorithmic}[1]
\STATE Let $j$ denote an arriving job.

\IF{there is sufficient green energy to schedule $j$}

\STATE employ First-Fit to schedule $j$;

\ELSE

\STATE use probability $p$ to schedule $j$ using First-Fit;

\STATE use probability $1 - p$ to schedule $j$ using Best-Fit.

\ENDIF
\end{algorithmic}
\label{alg_Random-Fit}
\end{algorithm}

The randomness (i.e., probability $p$) in Algorithm~\ref{alg_Random-Fit} varies for the cases in which job arrives at daytime or at nighttime. We calculate the \emph{best} value for $p$ in the following analysis.

Note that the algorithm Random-Fit is a probabilistic linear combination of First-Fit and Best-Fit. Thus, for the most general case, Random-Fit's competitive ratio is no better than the minimum lower bound of those of First-Fit and Best-Fit. In the following, we consider the special case in which all jobs are with the same lengths and node requirements.

\begin{theorem}
In scheduling jobs with the same processing times and the same node requirements, algorithm Random-Fit has its competitive ratio
\begin{displaymath}
c := \max \left\lbrace 1 + \frac{v_{on}}{v_{off}} - \left(\frac{v_{on}}{v_{off}}\right)^2, 1 + \frac{v_{off}}{v_g} - \left(\frac{v_{off}}{v_g}\right)^2 \right\rbrace,
\end{displaymath}
against an oblivious adversary. This competitive ratio $c$ is no more than $1.25$.
\label{comRatio_RF}
\end{theorem}

Before we proceed to the proof, we introduce some notation appearing in Theorem~\ref{comRatio_RF}. According to the definition of profit, a job $j$ with $p_j$ processing time and $q_j$ node requirement has profit
\begin{displaymath}
c \cdot p_j \cdot q_j - \int_t P(t) \cdot q_j,
\end{displaymath}
where $P(t)$ is the average unit energy price at time $t$, and $P(t)$ has the value $0$ (for green energy), $B^d$ (for on-peak brown energy), or $B^n$ (for off-peak brown energy) respectively when the job is processed using various types of energy. $P(t)$ is in integral along the time when the machines process $j$. If all jobs are with the same processing time and node requirements, then we normalize the profit as
\begin{displaymath}
\frac{c \cdot p_j \cdot q_j - \int_t P(t) \cdot q_j}{c \cdot p_j \cdot q_j} = 1 - \int_t \frac{P(t)}{c \cdot p_j }.
\end{displaymath}

In our proofs below, for ease to present the competitive ratio, we define $1 - \int_t \frac{P(t)}{c \cdot p_j}$ as $v_{on}$, $v_{off}$, $v_g$ as below.

\begin{displaymath}
1 - \int_t \frac{P(t)}{c \cdot p_j } :=
\begin{cases}
v_{on}, & \text{if only using on-peak brown energy to schedule $j$}\\
v_{off}, & \text{if only using off-peak brown energy to schedule $j$}\\
v_g, & \text{if only using green energy to schedule $j$}
\end{cases}
\end{displaymath}

Note that the normalized profit  $1 - \int_t \frac{P(t) }{c \cdot p_j }$ has a value among $(0, 1]$. According to the fact that on-peak brown energy is expensive than off-peak brown energy. Also, green energy has cost $0$. We have $0 < v_{on} < v_{off} < v_{g} = 1$. Also, for jobs with the same processing times and same node requirements, they have the same value for $v_{on}$, $v_{off}$, and $v_g$.

\begin{proof}[of Theorem~\ref{comRatio_RF}]
Let OPT denote an optimal offline algorithm (an oblivious adversary) as well as its net profit. Let RF denote the Random-Fit algorithm as well as its expected net profit. In order to prove this theorem, we will show that $OPT / RF \le 1.25$.

Our analysis consists of two algorithmic techniques: (1) We consider a special setting such that for this setting, the ratio of net profit of OPT and RF is no better than $c$ specified in Theorem~\ref{comRatio_RF}. This algorithmic technique has been used in proving a speed-scaling algorithm's competitive ratio by Yao \emph{et al.} in~\cite{YaoDS95}; (2) We employ a charging scheme such that at any time, OPT's amortized gain (of net profit) is no more than $c$ times of the expected gain of RF. We shall define an invariant to show this charging scheme's correctness.

For the first algorithmic technique, we claim that for the more restricted case in which all jobs have their processing times equal to $1$, the competitive ratio is no better than $c$. The reason is as below: As all jobs are identical in machine resource requirements, then it does not hurt for OPT to schedule the earliest-released job in the time interval when it schedules a job. Later arrivals cannot preempt any job scheduled already. Thus, if for the more restricting setting in which $p_j = 1$, we have a competitive ratio $c'$, then we have $c' = c$.

For the second algorithmic technique, we employ a charging scheme to prove Theorem~\ref{comRatio_RF}. Initially, OPT and RF have the same energy resource and machine resource. We consider an arriving job at time $t_1$ with the inductive assumption that before time $t_1$, the ratio between OPT and RF is no more than $c$. In the following, we show that after time $t_1$, the inductive assumption still holds.

Two facts are used in the proof: (1) Randomness only plays its role when no green energy is available (otherwise, no random decision is needed, see Algorithm~\ref{alg_Random-Fit}); and (2) If OPT schedules a job at time $t$, then OPT schedules the earliest-deadline job as all jobs are with the same processing times and node requirements. We will show that the following invariant holds: At any time, the net profit ratio between OPT and RF is no more than $c$; also, OPT has no more remaining green energy than RF does, if we charge appropriate revenue to OPT. This includes the scenario in which OPT schedules a job later with energy consumption while we charge the revenue and the energy cost for now for OPT. Once this invariant holds, Theorem~\ref{comRatio_RF} holds immediately. We consider the release jobs via case study and use $(r, d)$ to denote a job with release time $r$ and deadline $d$.

\paragraph{Consider the two neighboring time slots $t_1$ and $t_2$ which are at on-peak and at off-peak respectively}

\begin{enumerate}
\item OPT releases one job $j_1 = (t_1, t_2)$ and OPT schedule $j_1$ at time $t_2$, achieving a profit $v_{off}$. While RF will schedule job $j_1$ to time $t_1$ with probability $p$ and to time slot $t_2$ with probability $1 - p$, earning an expected profit $p \cdot v_{on} + (1 - p) \cdot v_{off}$. In this case, the competitive ratio is $\frac{OPT}{RF} = \frac{v_{off}}{p \cdot v_{on} + (1 - p) \cdot v_{off}}$.

\item OPT releases two jobs $j_1 = (t_1, t_2)$ and $j_2 = (t_2, t_2)$. OPT would schedule $j_1$ at time $t_1$ and schedule $j_2$ at time $t_2$, achieving a profit of $v_{on} + v_{off}$. While, the RF will schedule $j_1$ at time $t_1$ with probability $p$ and schedule either $j_1$ or $j_2$ at time $t_2$ (due to the job deadline constraints), earning a profit of $p \cdot v_{on} + v_{off}$. Therefore, the competitive ratio is $\frac{OPT}{RF} = \frac{v_{on} + v_{off}}{p \cdot v_{on} + v_{off}}$.
\end{enumerate}

In this scenario, the competitive ratio is:
\begin{displaymath}
\min_p\{\max\{\frac{v_{off} }{p \cdot v_{on} + (1 - p) \cdot v_{off}}, \frac{v_{on} + v_{off} }{p \cdot v_{on} + v_{off}}\}\}
\end{displaymath}

In solving above min-max problem, $p = \frac{x}{1 + x - x^2}$ where $x = \frac{v_{on}}{v_{off}}$ optimizes the competitive ratio
\begin{displaymath}
\frac{OPT}{RF} = 1 + x - x^2 = 1 + \frac{v_{on}}{v_{off}} - \left(\frac{v_{on}}{v_{off}}\right)^2 \le 1.25
\end{displaymath}

\paragraph{Consider the two neighboring time slots $t_1$ and $t_2$ which are at off-peak and at on-peak respectively}

\begin{enumerate}
\item OPT releases one job $j_1 = (t_1, t_2)$. The worst-case is that OPT uses the on-peak day's free green energy to schedule this job $j_1$. Using the same analysis approach, we get a competitive ratio $\frac{OPT}{RF} = \frac{v_g}{p \cdot v_{off} + (1 - p) \cdot v_g}$.

\item OPT releases two jobs $j_1 = (t_1, t_2)$ and $j_2 = (t_2, t_2)$. Similarly, we get a competitive ratio $\frac{OPT}{RF} = \frac{v_{off} + v_g}{p \cdot v_{off} + v_g}$.
\end{enumerate}

In this scenario, the competitive ratio is
\begin{displaymath}
\min_p \max\left\lbrace \frac{v_g}{p \cdot v_{off} + (1-p) \cdot v_g} , \frac{ v_g + v_{off} }{p \cdot v_{off} + v_g}\right\rbrace
\end{displaymath}

Similarly, when $p = \frac{y}{1 - y - y^2}$ where $y = \frac{v_{off}}{v_g}$ ($0 < y < 1$), we get the optimal competitive ratio
\begin{displaymath}
\frac{OPT}{RF} = 1 + y - y^2 = 1 + \frac{v_{off}}{v_g} - \left(\frac{v_{off}}{v_g}\right)^2  \leq 1.25
\end{displaymath}
\end{proof}

\begin{corollary}
Random-Fit has better competitive ratio compared to First-Fit and Best-Fit.
\end{corollary}

\begin{proof}
We use FF and BF to stand for First-Fit and Best-Fit. As the expected competitive ratio of Random-Fit is no larger than $\max\{1 + \frac{v_{on}}{v_{off}} - \left(\frac{v_{on}}{v_{off}}\right)^2, 1 + \frac{v_{off}}{v_g} - \left(\frac{v_{off}}{v_g}\right)^2\}$, we have $\frac{OPT}{RF} < \frac{OPT}{BF}$. Since $1 + k - k^2 < 1 / k$ (where $0 < k <1$), we have $1 + \frac{v_{on}}{v_{off}} - \left(\frac{v_{on}}{v_{off}}\right)^{2} < \frac{v_{off}}{v_{on}}$, and $1 + \frac{v_{off}}{v_g} - \left(\frac{v_{off}}{v_g}\right)^2  < \frac{v_g}{v_{off}}$, then we have $\frac{OPT}{RF} < \frac{OPT}{FF}$.
\end{proof}

\begin{corollary}
The optimal randomness (probability) for Random-Fit is
\begin{displaymath}
\begin{cases}
p = \frac{x}{1 + x - x^2}, & x = \frac{v_{on}}{v_{off}}\\
p' = \frac{y}{1 - y - y^2}, & y = \frac{v_{off}}{v_g}
\end{cases}
\end{displaymath}
for scheduling jobs from on-peak time to off-peak time and from off-peak time to on-peak time respectively.
\end{corollary}


\section{Performance Evaluation}

In this section, we evaluate the randomized online algorithm Random-Fit against two deterministic online algorithms First-Fit and Best-Fit which have been revised and adopted in previous literature. An offline algorithm is also developed, though its running time is tedious when the input size is large. The algorithms are implemented under both the job preemption setting and job non-preemption setting. For ease of presentation in the figures below, we abbreviate the First-Fit algorithm, the Best-Fit algorithm, the Random-Fit algorithm and their preemptive versions as: FF, BF, RF, PFF, PBF, and PRF, respectively.


\subsection{Experiment settings}

\paragraph*{Data center}

The simulated green data center is configured the same as the one in~\cite{GoiriL11}. The data center is a cluster consisting of $16$ nodes with each node consumes $140$W when they are running jobs (we call them active) and  $0$W otherwise. The total energy consumption is the sum of the energy consumed by the nodes being active over time.

\paragraph*{Green energy}

We use the solar energy trace from the Computer Science Weather Station at University of Massachusetts, Amherst~\cite{solarTrace}. The solar energy is fine grained such that it is collected every $5$ minutes. We sum up the solar energy of each consecutive $3$ time periods, say $15$ minutes in total, to represent the solar energy at one time slot in our model. We scale down the solar power energy trace and make it compatible with the simulated data center by making the peak solar power cover $75\%$ of the maximum possible power consumption.

To fully evaluate our job scheduling policies under different whether conditions, we choose three types of days with ``high'', ``medium'', and ``low'' solar energy production. We run simulations under these three types of whether settings respectively. Note that for each type of days, we select an arbitrary $5$-day-period time. We only present the results when energy is ``high'' as we get similar results under these $3$ solar energy settings. A solar trace of ``high'' days is shown in Figure~\ref{fig_Solar_high}.

\begin{figure}[!ht]
\centering
\includegraphics[width=4in]{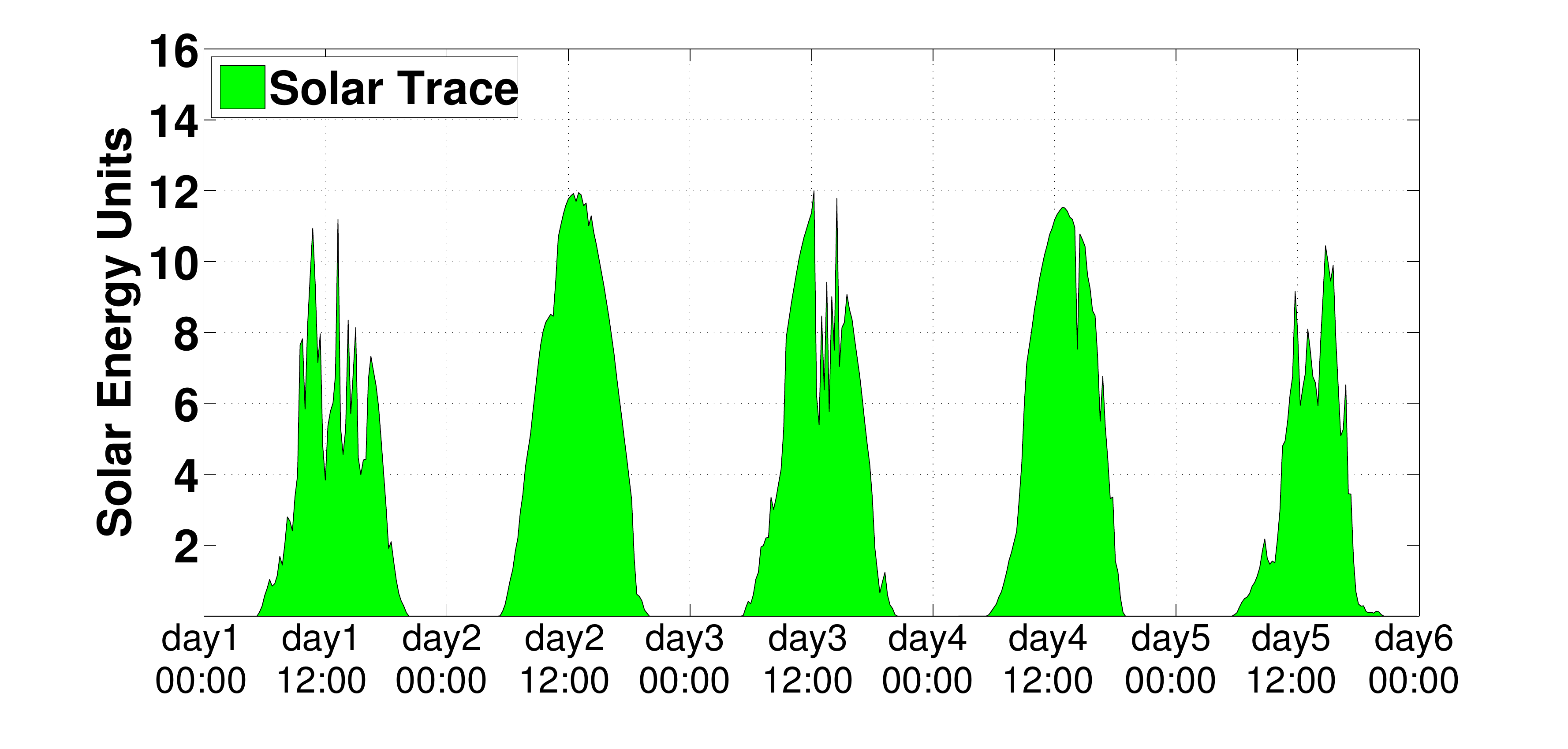}
\caption{Solar energy trace for arbitrary $5$ days.}
\label{fig_Solar_high}
\end{figure}

\paragraph*{Brown energy price}

We assume that the brown energy price is varying at on-peak/off-peak periods. The electricity cost is less at off-peak than at on-peak periods. We use the prices charged by PSEG in New Jersey at summer time~\cite{GoiriL11} as an example: on-peak price (from 9am to 11pm) $0.13/kWh$, off-peak price (from 11pm to 9 am) $0.08/kWh$.

\paragraph*{Service pricing}

The green data center service provider charges the clients for the computing resource consumed. We set the service price based on Amazon EC2' pricing~\cite{amazonprice}. The charging price is set as $\$0.022/h$ per machine.


\paragraph*{Workloads}

As the performance of online scheduling algorithms are sensitive to the workload sequences (as what we shall see), we simulate multiple types of workload traces in order to thoroughly evaluate the performance of online algorithms. Particularly, Random-Fit has its randomness factor internal to the algorithm and we need to conduct multiple rounds of experiments to find out the expected performance. In our simulations, we simulate $6$ types of workload traces. In the following, we give a detailed description of each workload. For ease of expression, we denote them as \textbf{RealTrace}, \textbf{UUTrace}, \textbf{UETrace}, \textbf{PUTrace}, \textbf{PETrace}, \textbf{StaggeredTrace} respectively.

\begin{itemize}
\item \textbf{RealTrace} --- For this real workload trace, we use Grid5k~\cite{Grid5k} which is collected from Grid'5000 system~\cite{Grid5000Platform}, a $2218$ node experimental grid platform consisting of 9 sites geographically distributed in France, from May 2004 to November 2006. We select an arbitrary $5$-day-period set of jobs, which consists of $4269$ jobs. We randomly select a subset of jobs as the input of the simulation. The number of jobs chosen varies in simulating different workload sizes.

\item \textbf{UUTrace} --- In this trace, jobs have uniform arrival patterns, processing time requirements, and node requirements. Each job $j$ has an arrival time $r_j$ uniformly chosen from $[1, 480]$ (as there are $480$ time slots in a $5$-day period), required processing time $p_j$ uniformly chosen from $[1, 9]$, required node number $q_j$ uniformly chosen from $[1, 5]$, deadline $d_j$ uniformly chosen from $[r_j + p_j, 480]$.

\item \textbf{UETrace} --- In this trace, jobs have uniform arrival patterns and all jobs have the same processing times and node requirements. The setting of the arriving times and deadlines is the same as \textbf{UUTrace}. The required processing time is set as $5$ and required node numbers is set as $3$. Note that we also simulate other settings of processing times and node requirements and get similar results. Therefore, only one representative group of results is shown.

\item \textbf{PUTrace} --- In this trace, jobs have Poisson arrival patterns. The processing times and node requirements satisfy uniform distribution as that of \textbf{UUTrace}. The job arrival rate is tuned in order to produce different workload sizes.

\item \textbf{PETrace} --- In this trace, jobs have Poisson arrival patterns. All jobs have the same required processing times and required node numbers  as that of \textbf{UETrace}.

\item \textbf{StaggeredTrace} --- In this trace, jobs arrive periodically as described in~\cite{GoiriL11}. The required processing times and required node numbers of jobs satisfy a uniform distribution as that of \textbf{UUTrace}. We assume $75\%$ jobs arrive at daytime and $25\%$ jobs arrive at night. Each job has its span (the difference between its deadline and its arriving time) set as $2$ days.
\end{itemize}

We simulate a workload for $5$ days and plot the average generated workload trace with expected utilization  $80\%$ (except for \textbf{RealTrace}) in Figure~\ref{fig:workload}. As \textbf{RealTrace} contains jobs of various sizes and the sizes of jobs may not satisfy a uniform distribution, it is difficult to characterize the utilization. Therefore, we plot the workload trace of randomly selected $400$ jobs. The workload at each time slot in Figure~\ref{fig:workload} is the sum of the nodes requested by all jobs at that time slot.

\begin{figure*}[h!]
\centering
\subfigure[{UUTrace}]{\includegraphics[width=.32\textwidth]{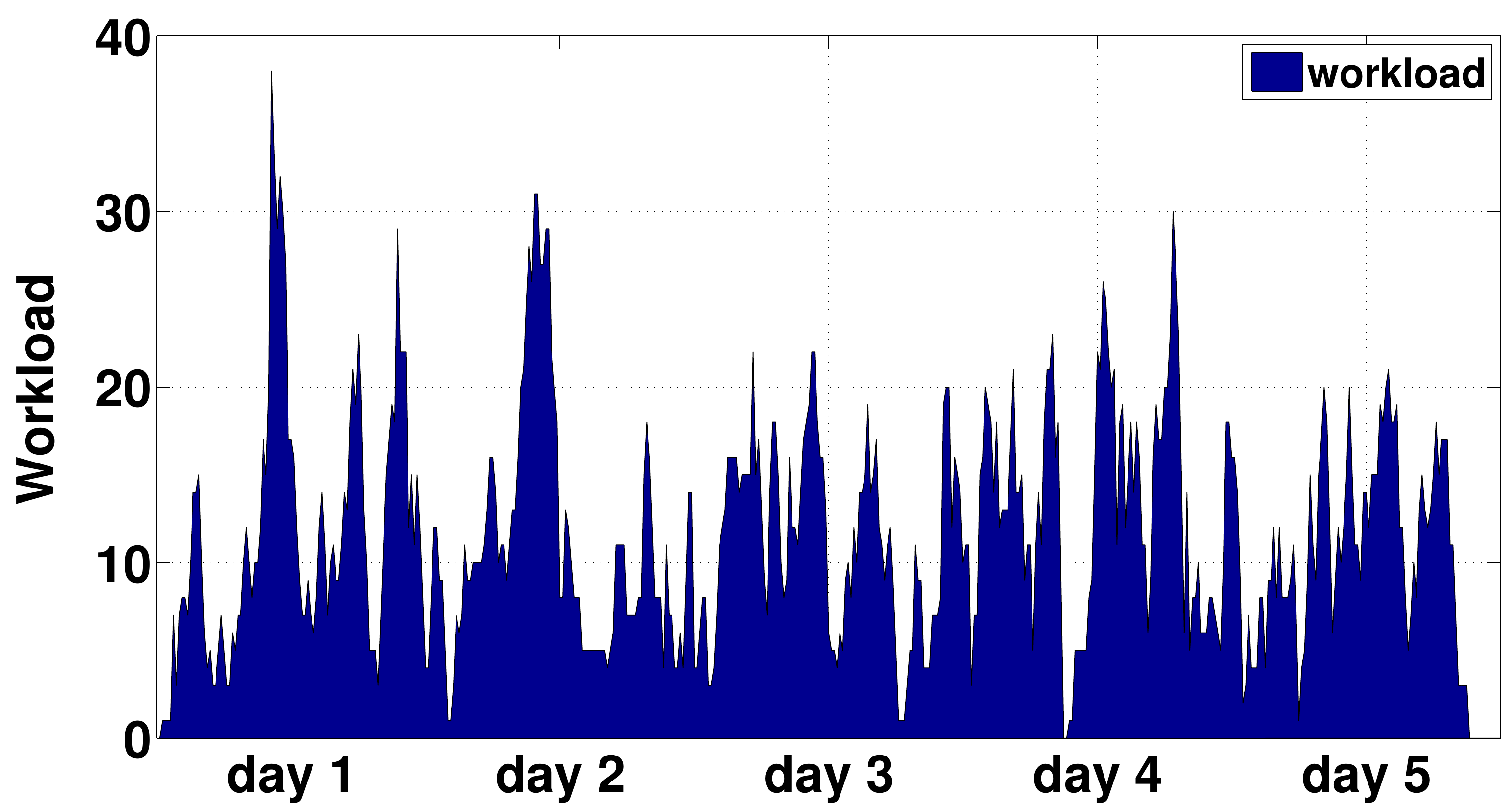}
\label{fig_Uniform_Uniform_wrokdLoad}}
\subfigure[{UETrace}]{\includegraphics[width=.32\textwidth]{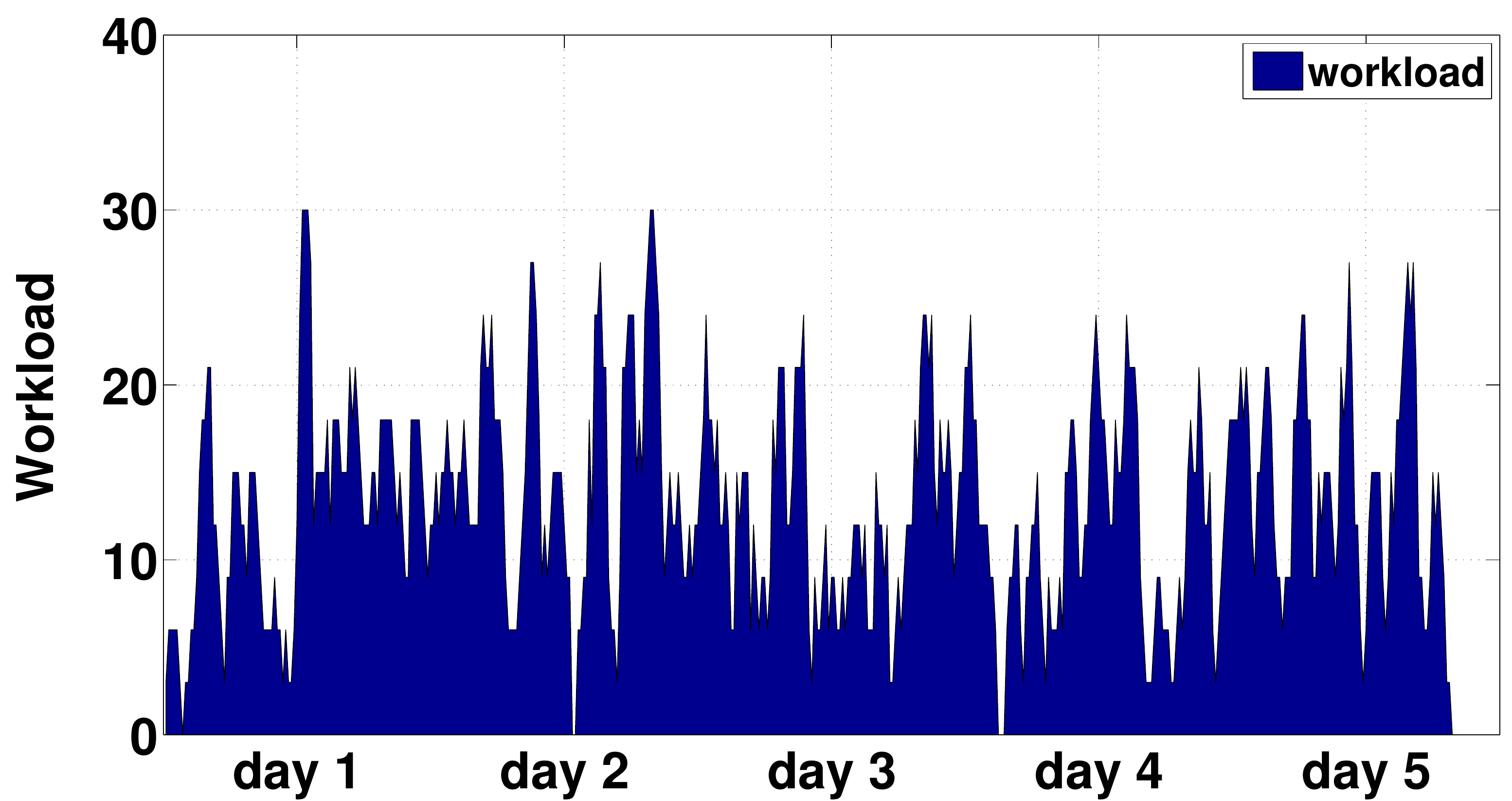}
\label{fig:Uniform_Equal_wrokdLoad}}
\subfigure[{PUTrace}]{\includegraphics[width=.32\textwidth]{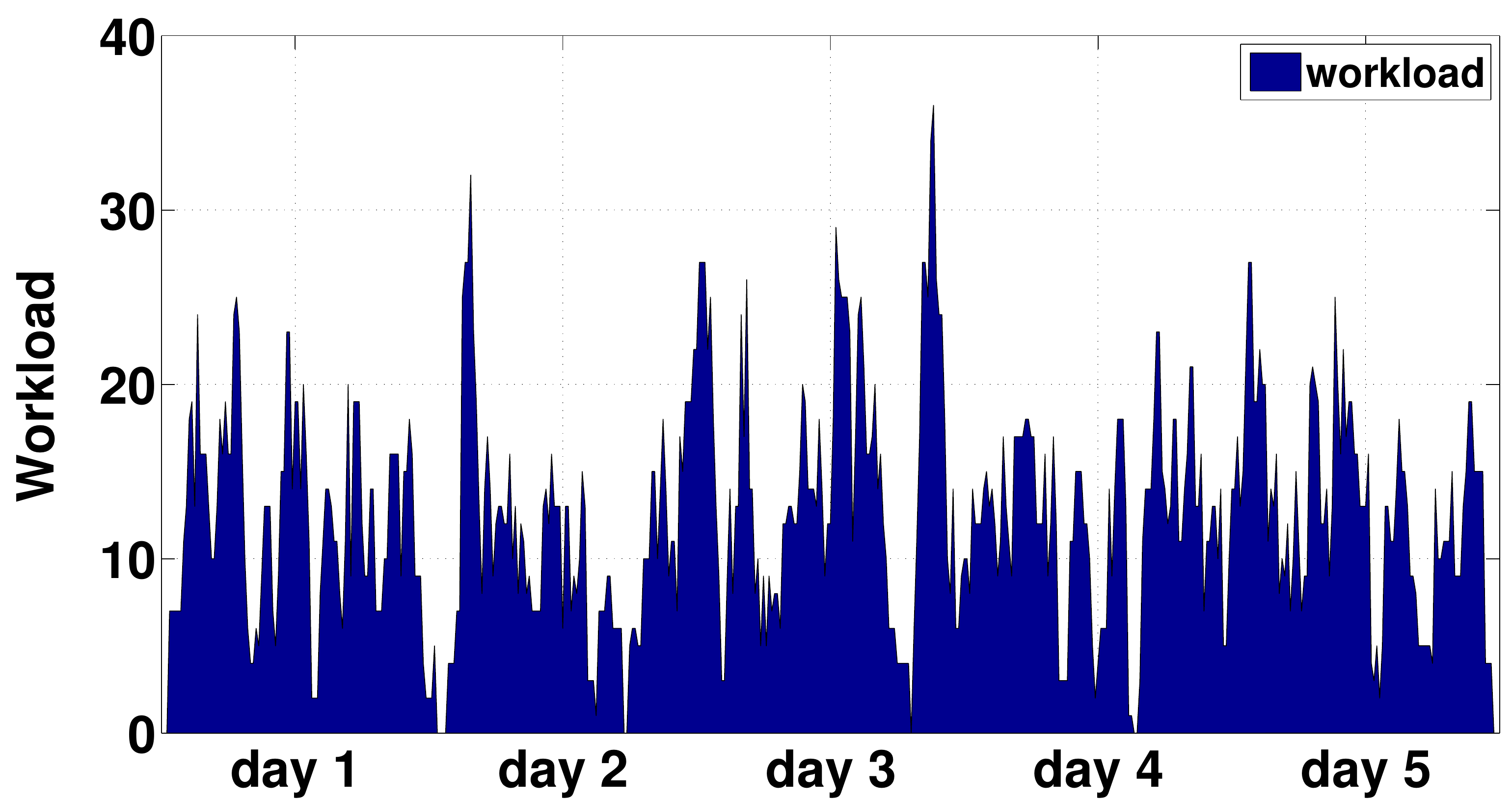}
\label{fig:Poisson_Uniform_wrokdLoad}}
\subfigure[{PETrace} ]{\includegraphics[width=.32\textwidth]{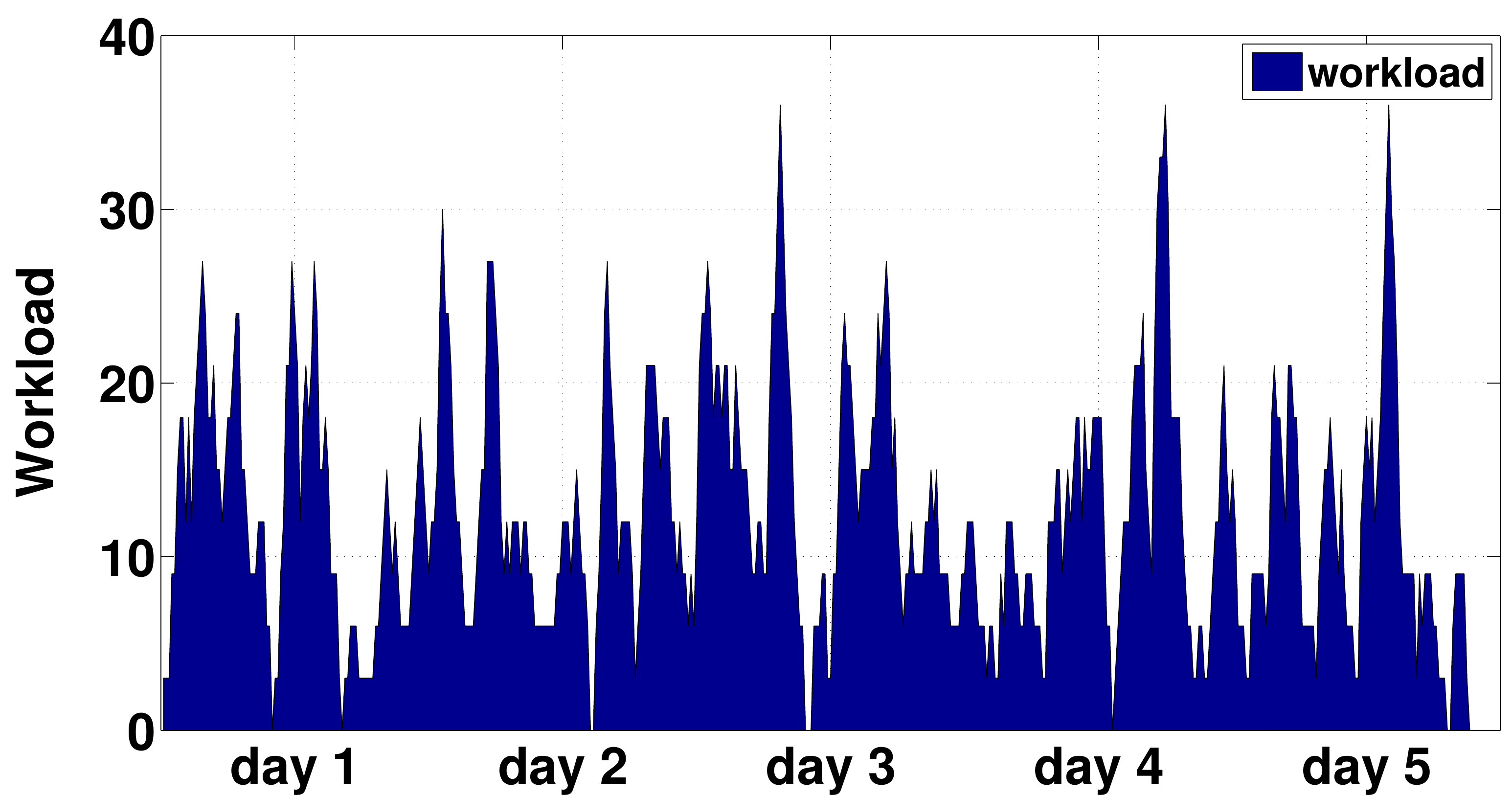}
\label{fig:Poisson_Equal_wrokdLoad}}
\subfigure[{StaggerdTrace}]{\includegraphics[width=.32\textwidth]{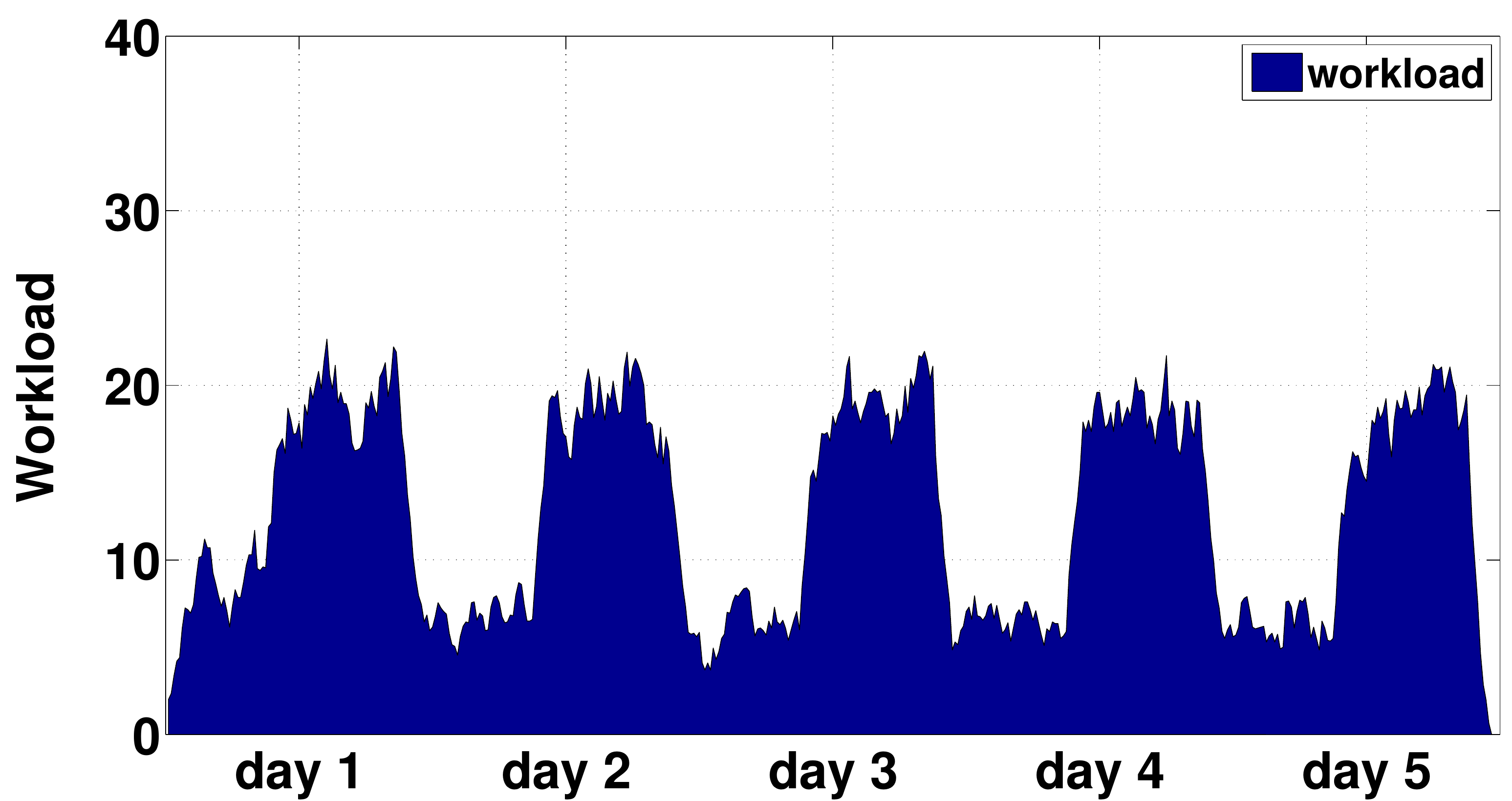}
\label{fig:Staggere_Uniform_wrokdLoad}}
\subfigure[{RealTrace}]{\includegraphics[width=.32\textwidth]{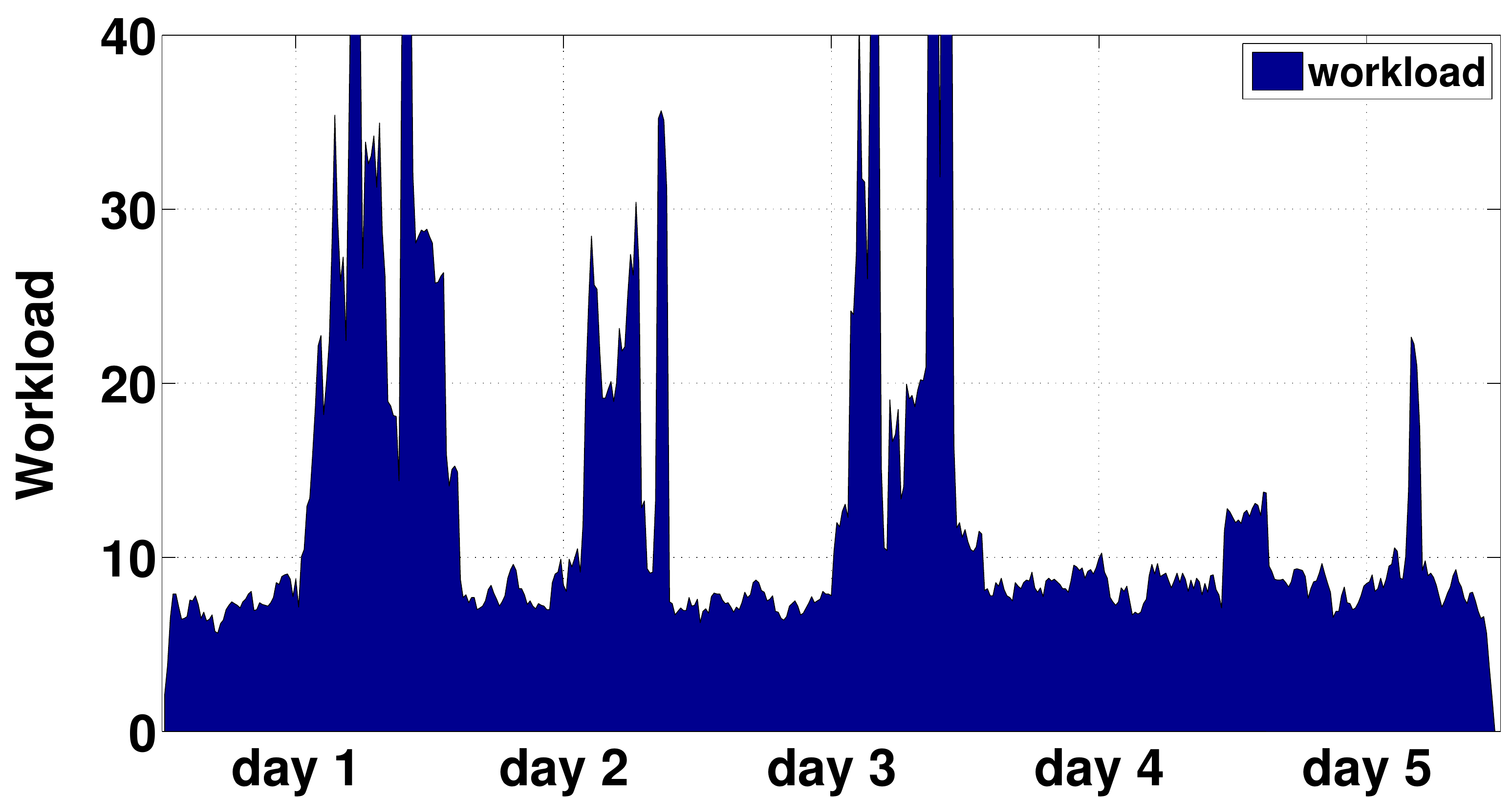}
\label{fig:Real_Uniform_wrokdLoad}}
\caption{Different types of workloads, all workloads have expected utilization  $80\%$ except for \textbf{RealTrace} which contains $400$ jobs.}
\label{fig:workload}
\end{figure*}


\subsection{Methodology}

We evaluate the performance of these online algorithms under various types of workloads. For each workload except \textbf{RealTrace}, we tune the workload utilization from $10\%$ to $150\%$. For \textbf{RealTrace}, we vary the number of jobs from $50$ to $750$. Under each setting, we compare the (1) \textbf{scheduled workload}, (2) \textbf{net profit}, (3) \textbf{green energy/brown energy consumption}. In comparison to the corresponding offline algorithms, we use \textbf{UETrace} workload setting since the theoretical bound has been validated under this setting when jobs are of the same processing times and node requirements. Each simulation is repeated for $30$ times and we compare the average value.


\subsection{Simulation results}

We first present the performance of First-Fit, Best-Fit and Random-Fit algorithms under various workload settings with the goal to show that Random-Fit guarantees a better worse-case performance. We take one step further to compare these algorithms with the optimal offline algorithm and confirmed our theoretical proof that Random-Fit has a better competitive ratio than First-Fit and Best-Fit. Then we demonstrate the performance of the these algorithms when job preemption is allowed. We will show that, counter our intuition, job preemption is not necessary to guarantee a higher revenue for all scheduling algorithms. In other words, job preemption may not help in maximizing net profit.


\subsubsection{Comparison of three online algorithms with no job preemption}
\label{subsubsec_simOnline}

The normalized profits under all the $6$ workload settings are shown in Figure~\ref{fig:profit}. Figure~\ref{fig:workload_scheduled} shows the workload scheduled. Figure~\ref{fig:usedGreenEnergy} and Figure~\ref{fig:usedBrownEnergy} show the green energy and the brown energy consumption respectively. To save space, we do not present all the results under every workload setting.

\begin{figure*}[h!]
\centering
\subfigure[]{\includegraphics[width=.32\textwidth]{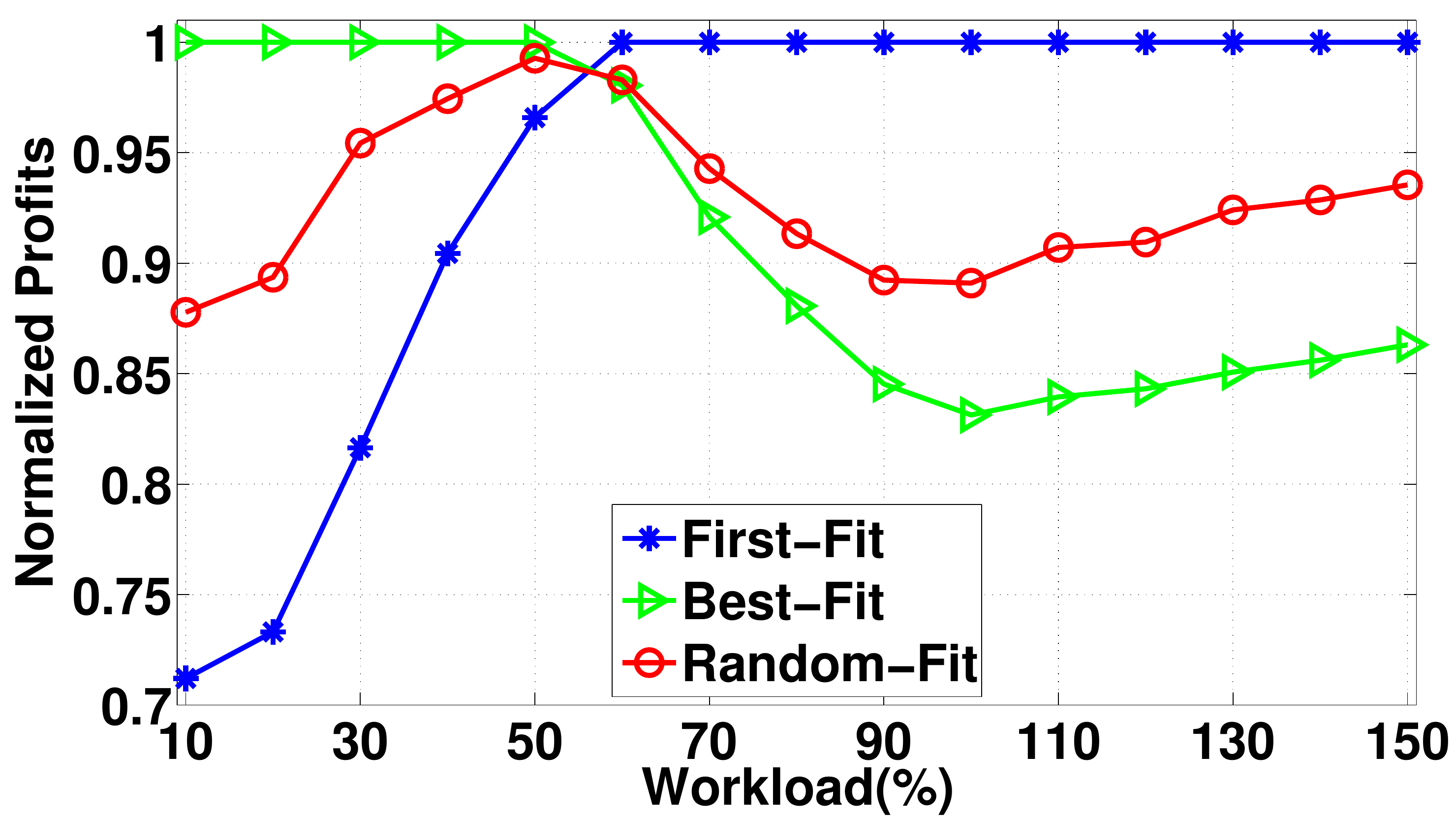}
\label{fig:profit_uniform_uniform}}
\subfigure[]{\includegraphics[width=.32\textwidth]{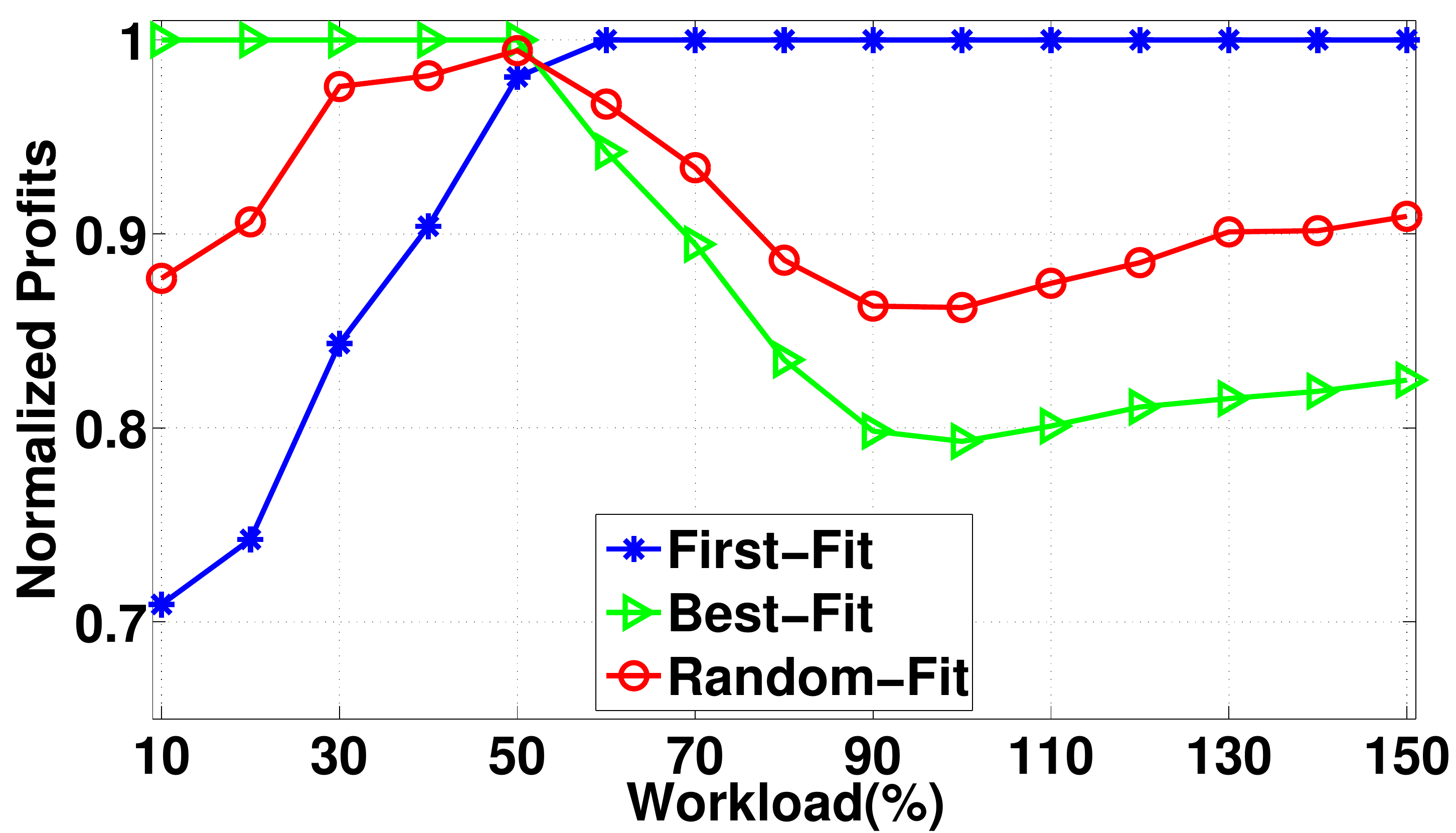}
\label{fig:profit_uniform_equal}}
\subfigure[]{\includegraphics[width=.32\textwidth]{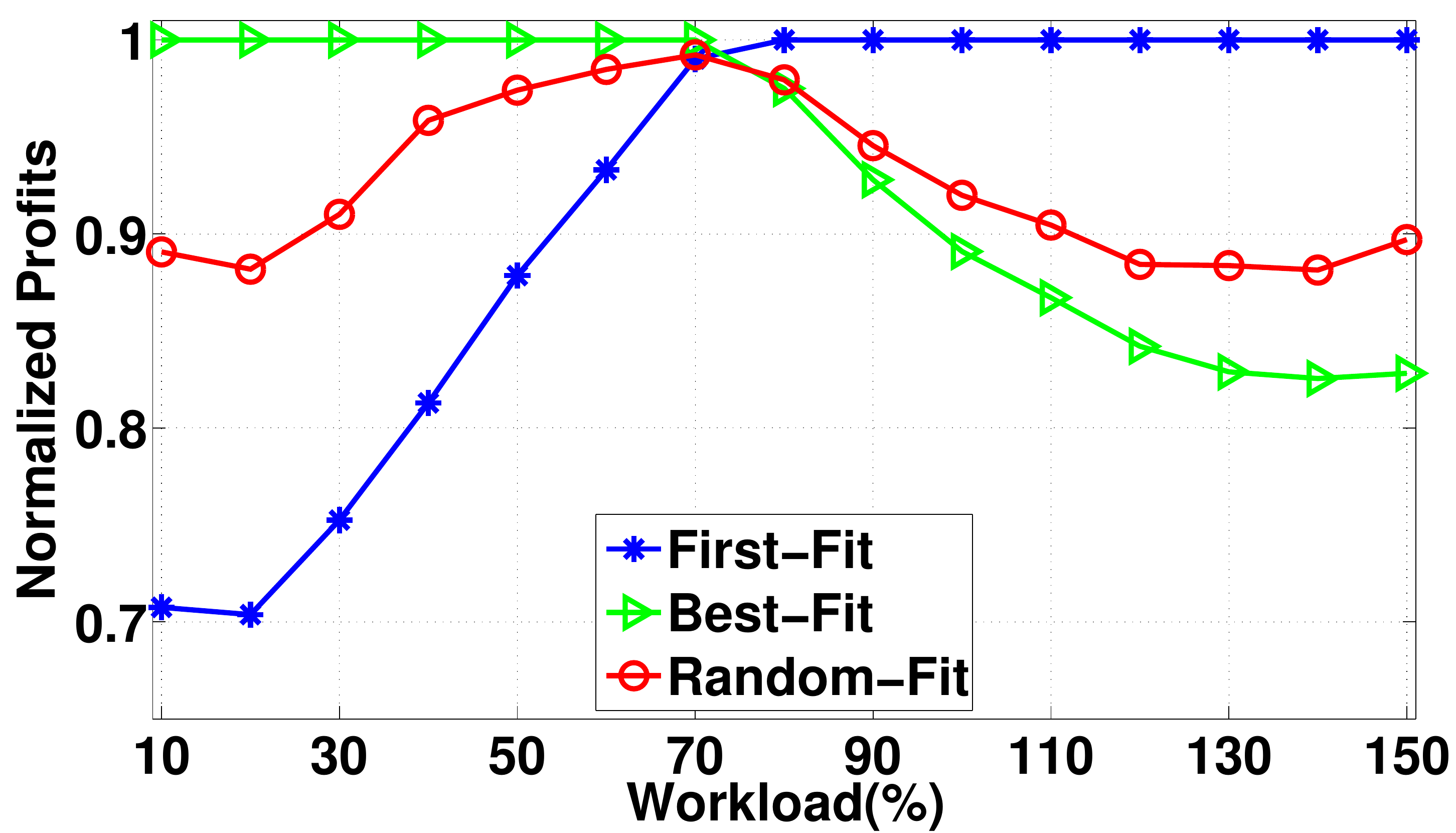}
\label{fig:profit_poisson_uniform}}
\subfigure[]{\includegraphics[width=.32\textwidth]{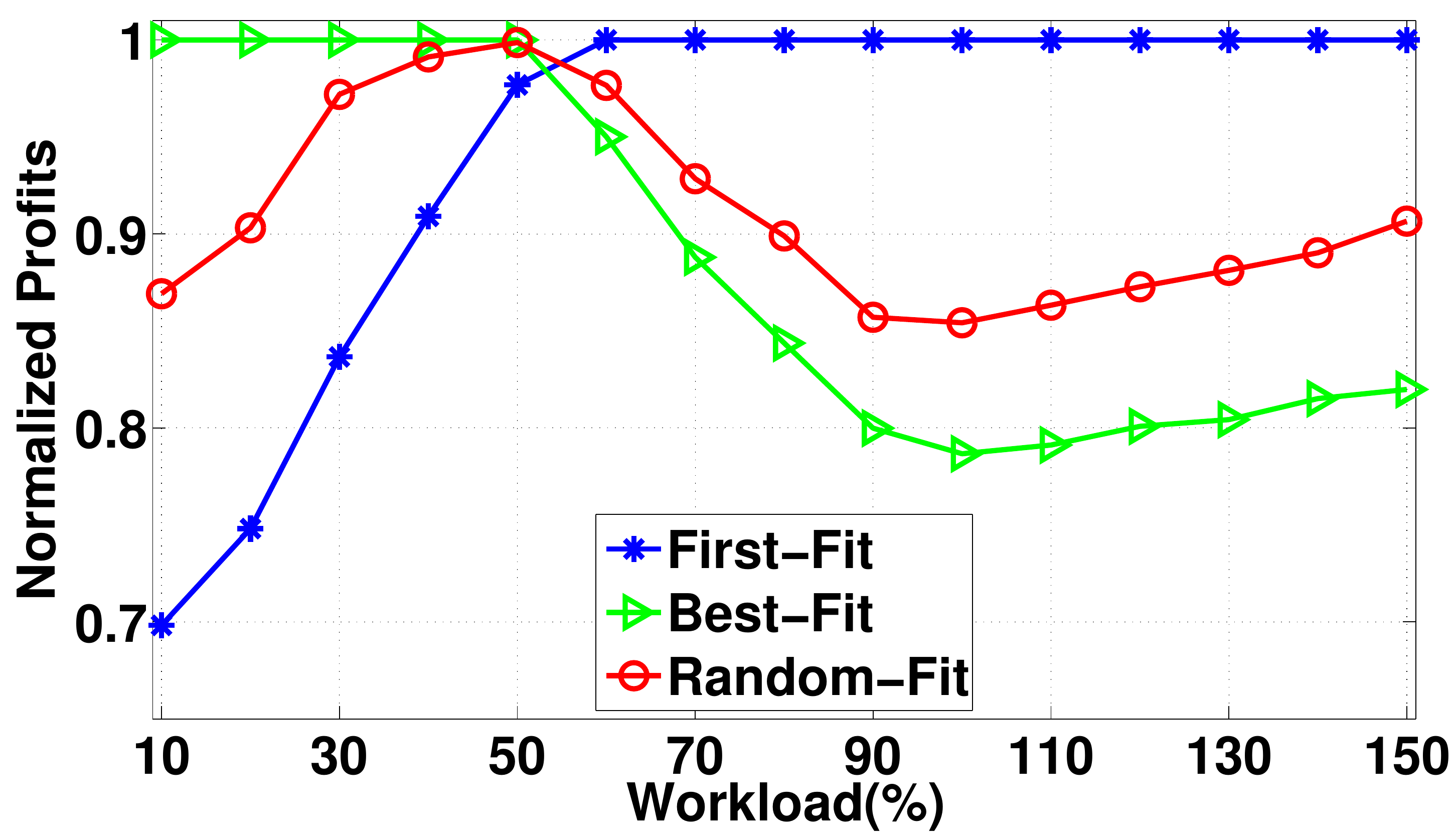}
\label{fig:profit_poisson_equal}}
\subfigure[]{\includegraphics[width=.32\textwidth]{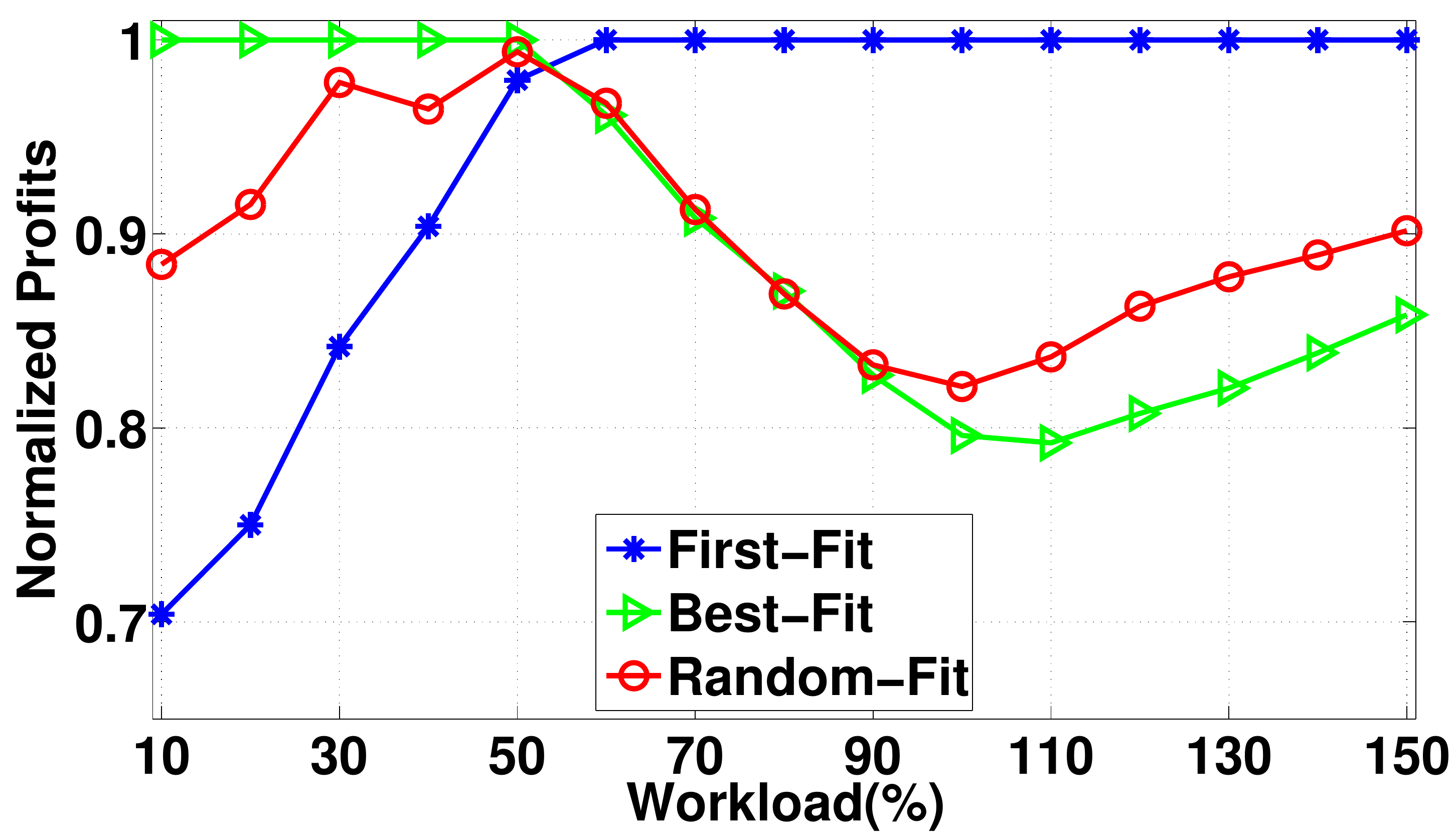}
\label{fig:profit_staggered_uniform}}
\subfigure[]{\includegraphics[width=.32\textwidth]{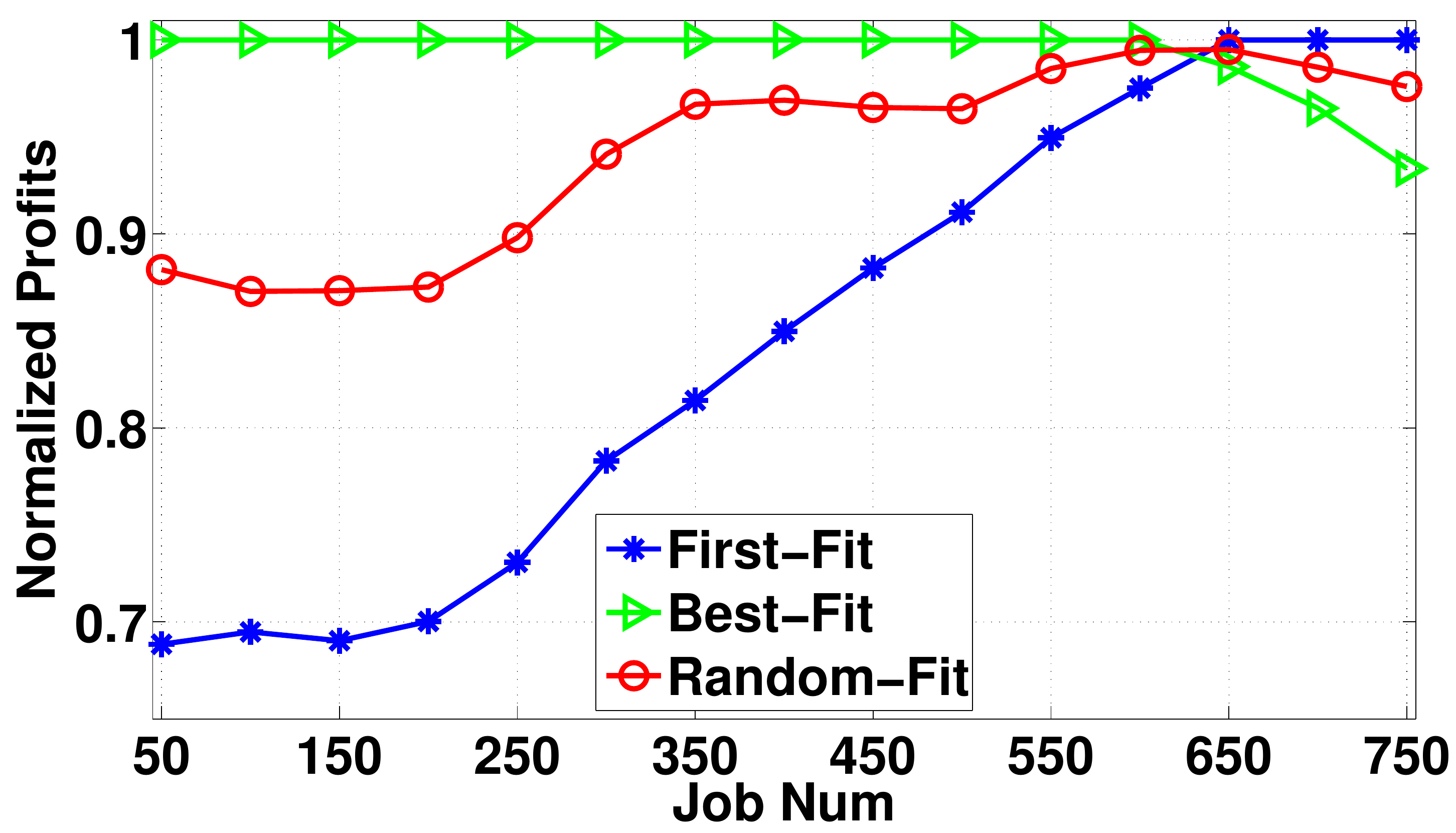}
\label{fig:profit_real}}
\caption{
Normalized profits under different workloads.
(a) UUTrace.
(b) UETrace.
(c) PUTrace.
(d) PETrace.
(e) StaggeredTrace.
(f) RealTrace.}
\label{fig:profit}
\end{figure*}

\begin{figure*}[h!]
\centering
\subfigure[]{\includegraphics[width=.32\textwidth]{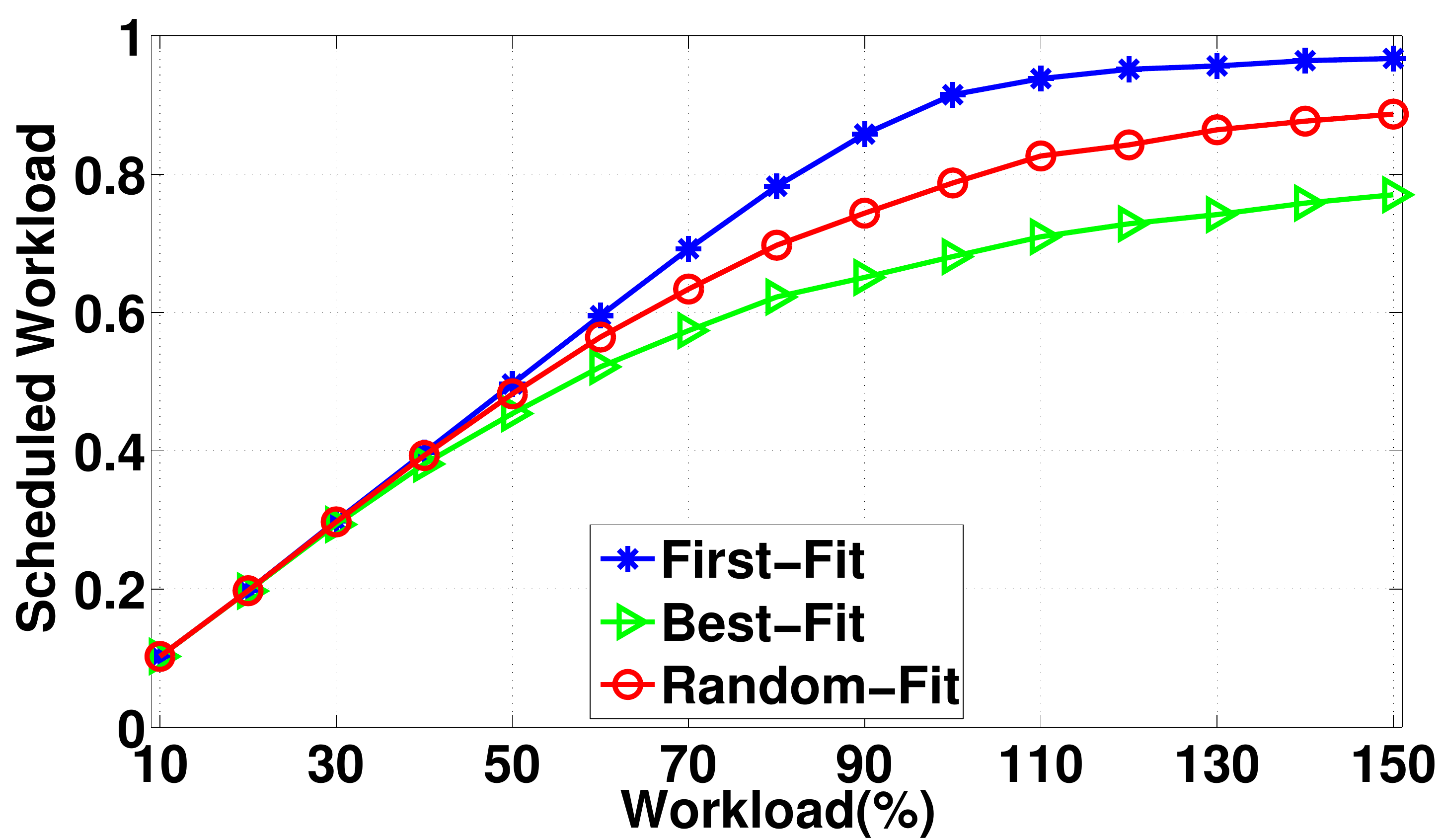}
\label{fig_workload_uniform}}
\subfigure[]{\includegraphics[width=.32\textwidth]{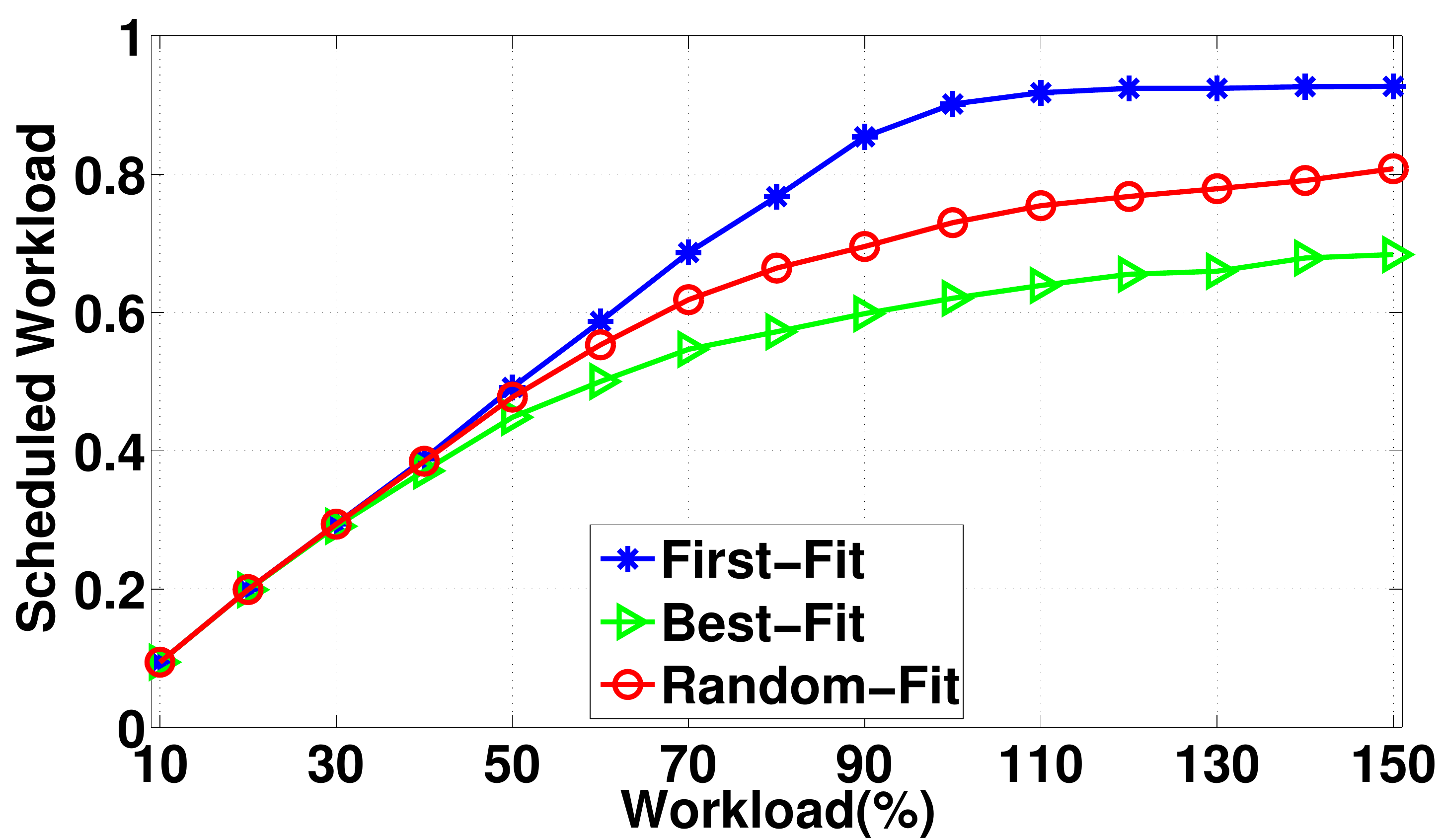}
\label{fig:workload_poisson_equal}}
\subfigure[]{\includegraphics[width=.32\textwidth]{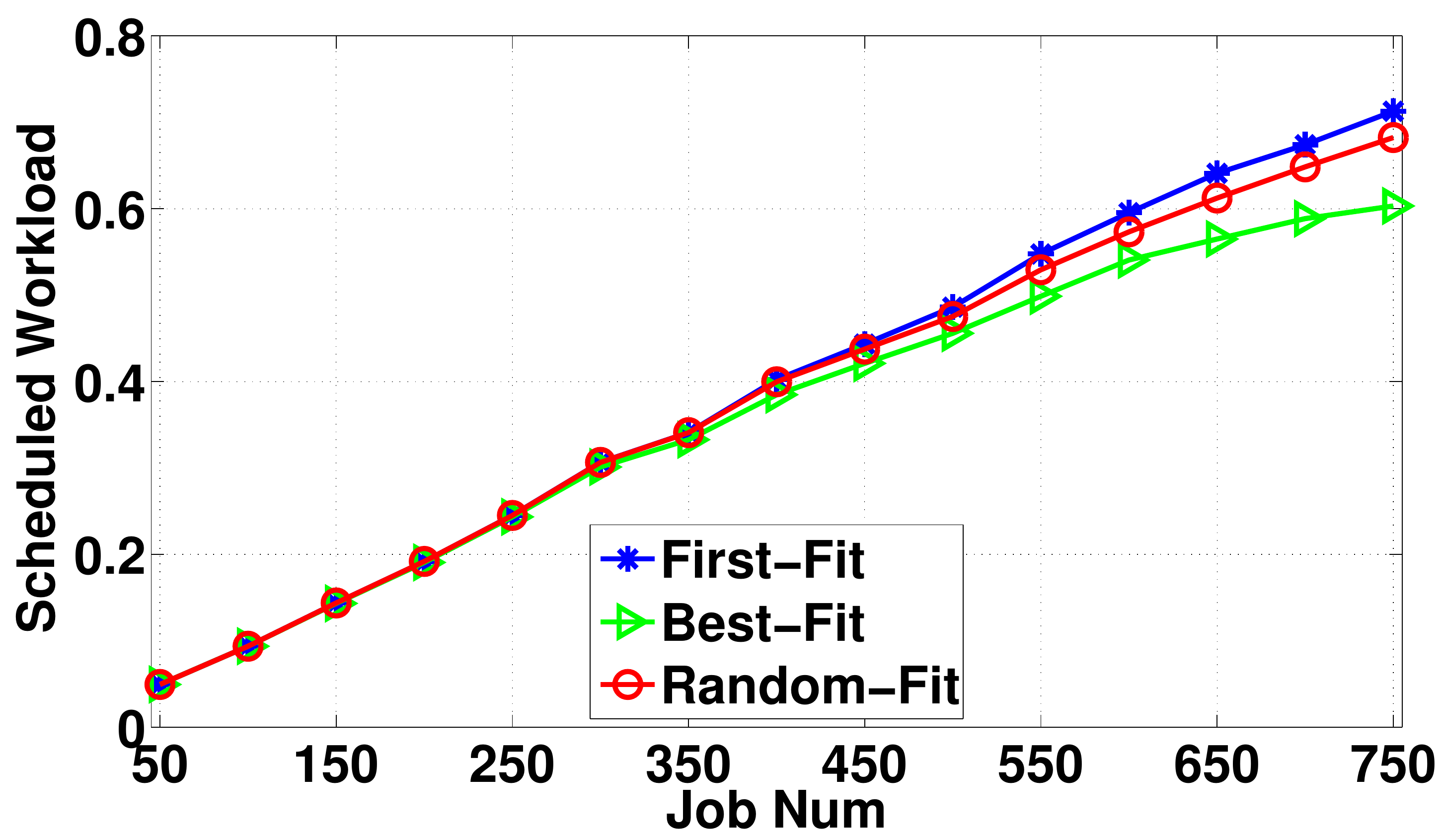}
\label{fig:workload_real}}
\caption{
Scheduled workloads under different workloads.
(a) UUTrace.
(b) PETrace.
(c) RealTrace.}
\label{fig:workload_scheduled}
\end{figure*}

\begin{figure*}[h!]
\centering
\subfigure[]{\includegraphics[width=.32\textwidth]{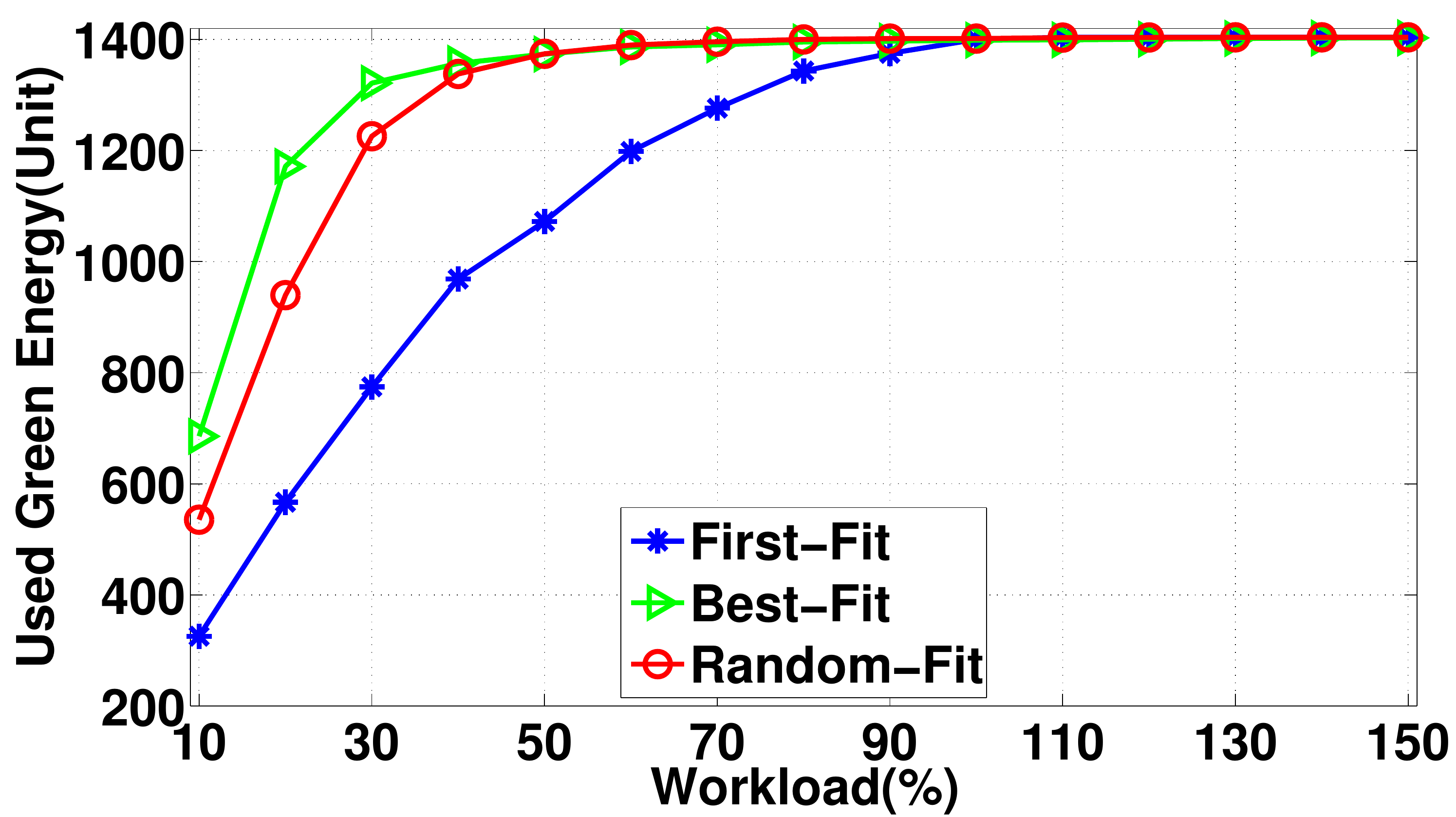}
\label{fig:UsedGreenEnergy_uniform}}
\subfigure[]{\includegraphics[width=.32\textwidth]{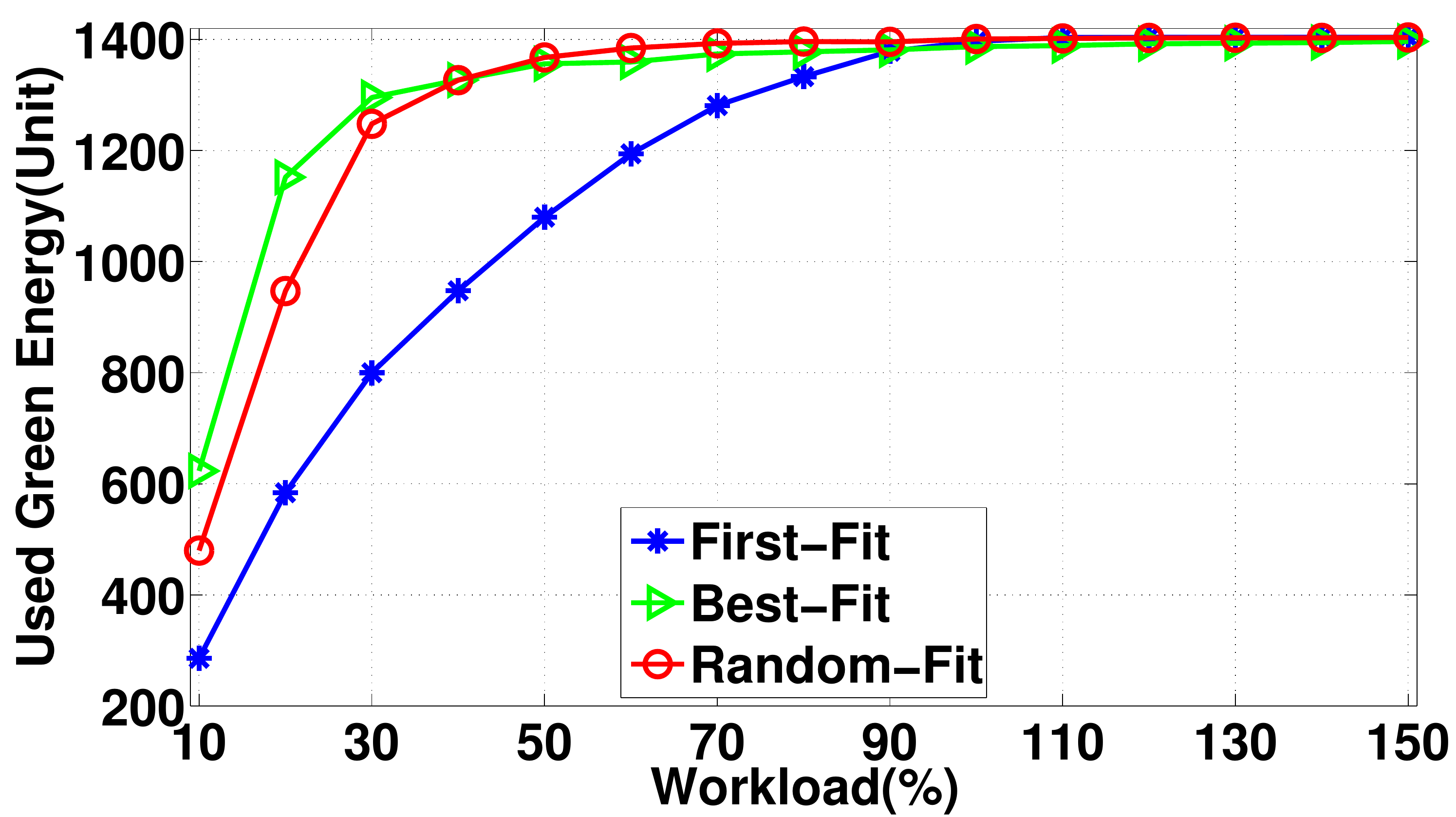}
\label{fig:UsedGreenEnergy_poisson_equal}}
\subfigure[]{\includegraphics[width=.32\textwidth]{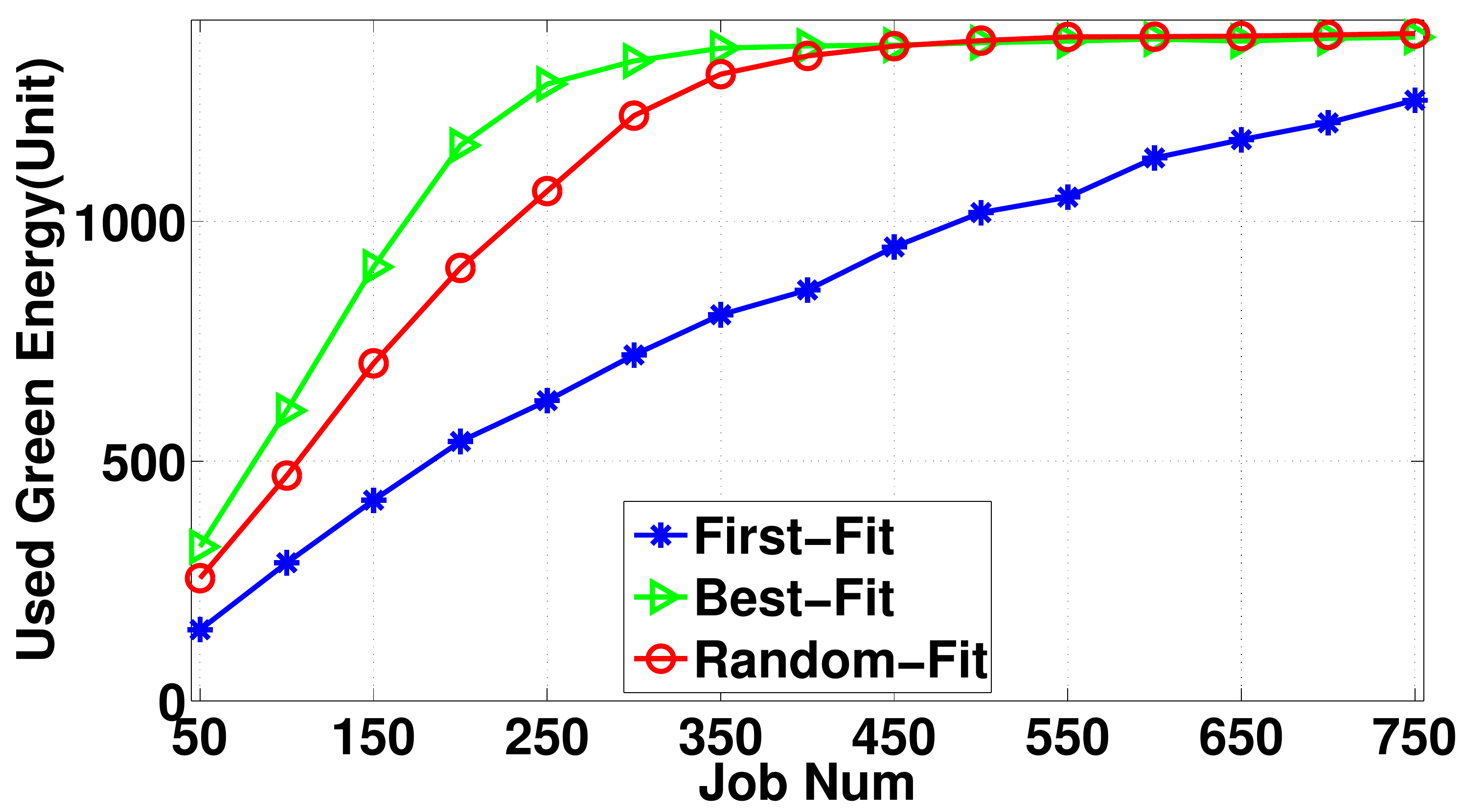}
\label{fig:UsedGreenEnergy_real}}
\caption{
Green energy consumption under different workloads.
(a) UUTrace.
(b) PETrace.
(c) RealTrace.}
\label{fig:usedGreenEnergy}
\end{figure*}

\begin{figure*}[h!]
\centering
\subfigure[]{\includegraphics[width=.32\textwidth]{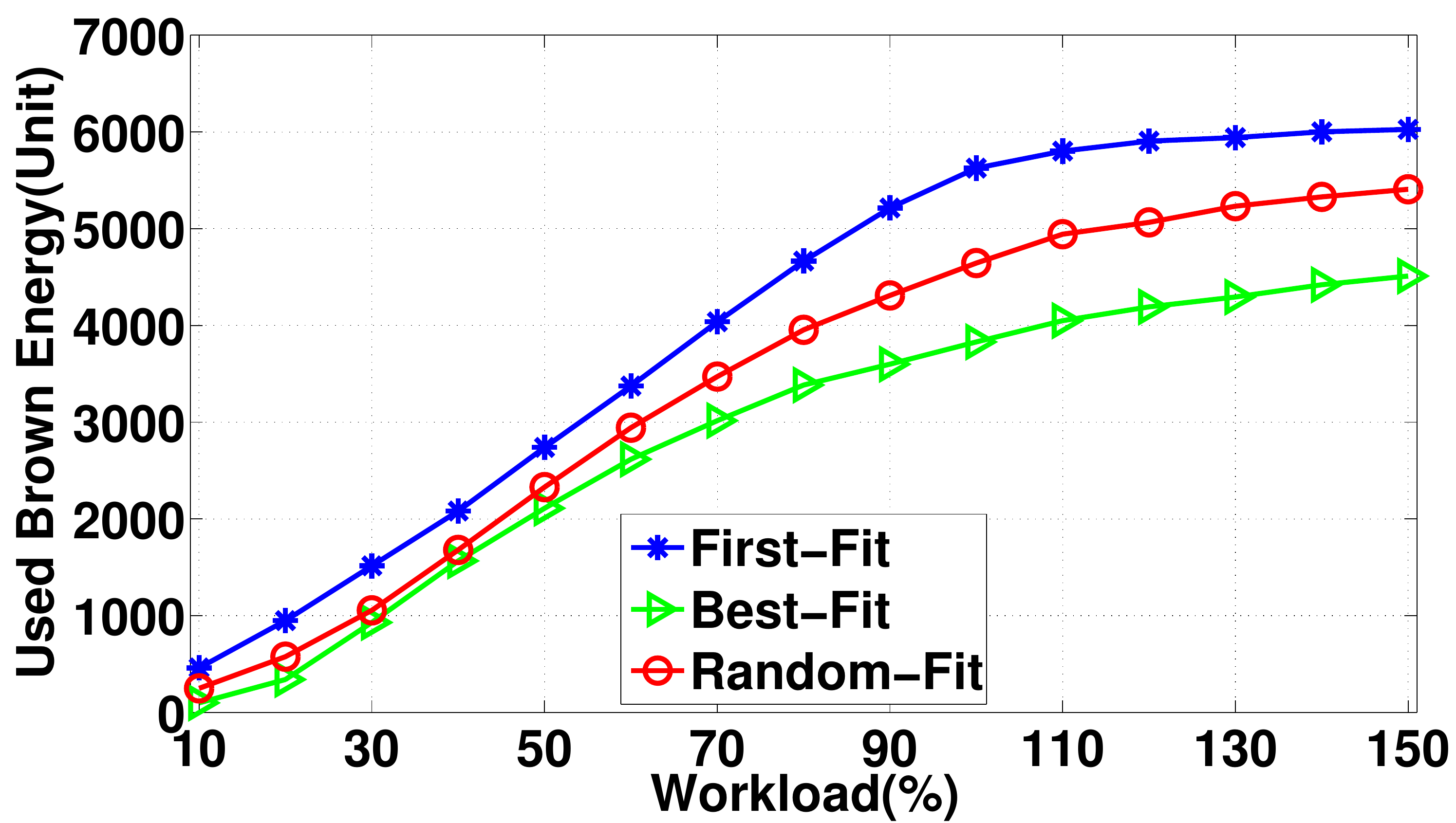}
\label{fig:UsedBrownEnergy_uniform}}
\subfigure[]{\includegraphics[width=.32\textwidth]{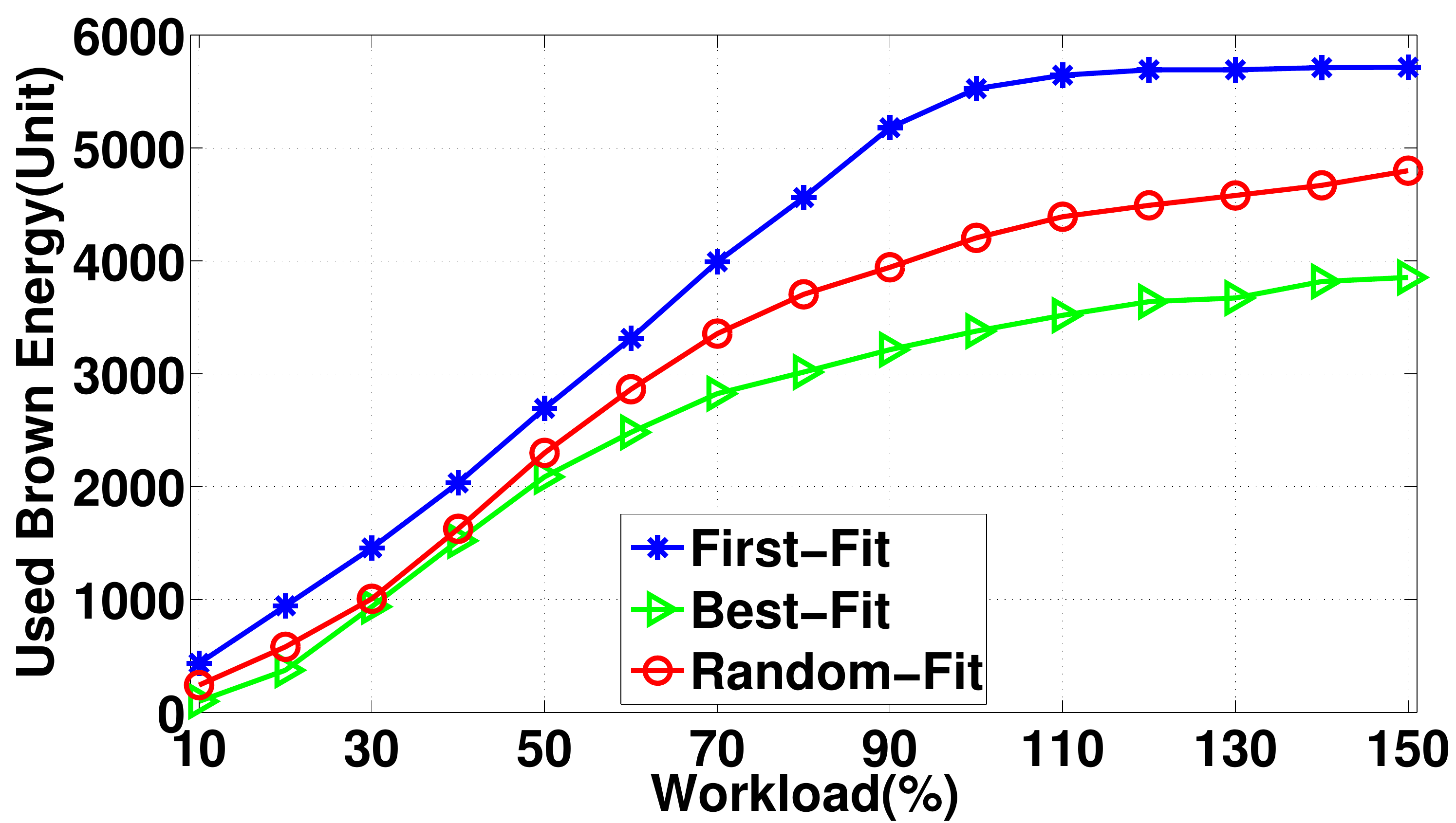}
\label{fig:UsedBrownEnergy_poisson_equal}}
\subfigure[]{\includegraphics[width=.32\textwidth]{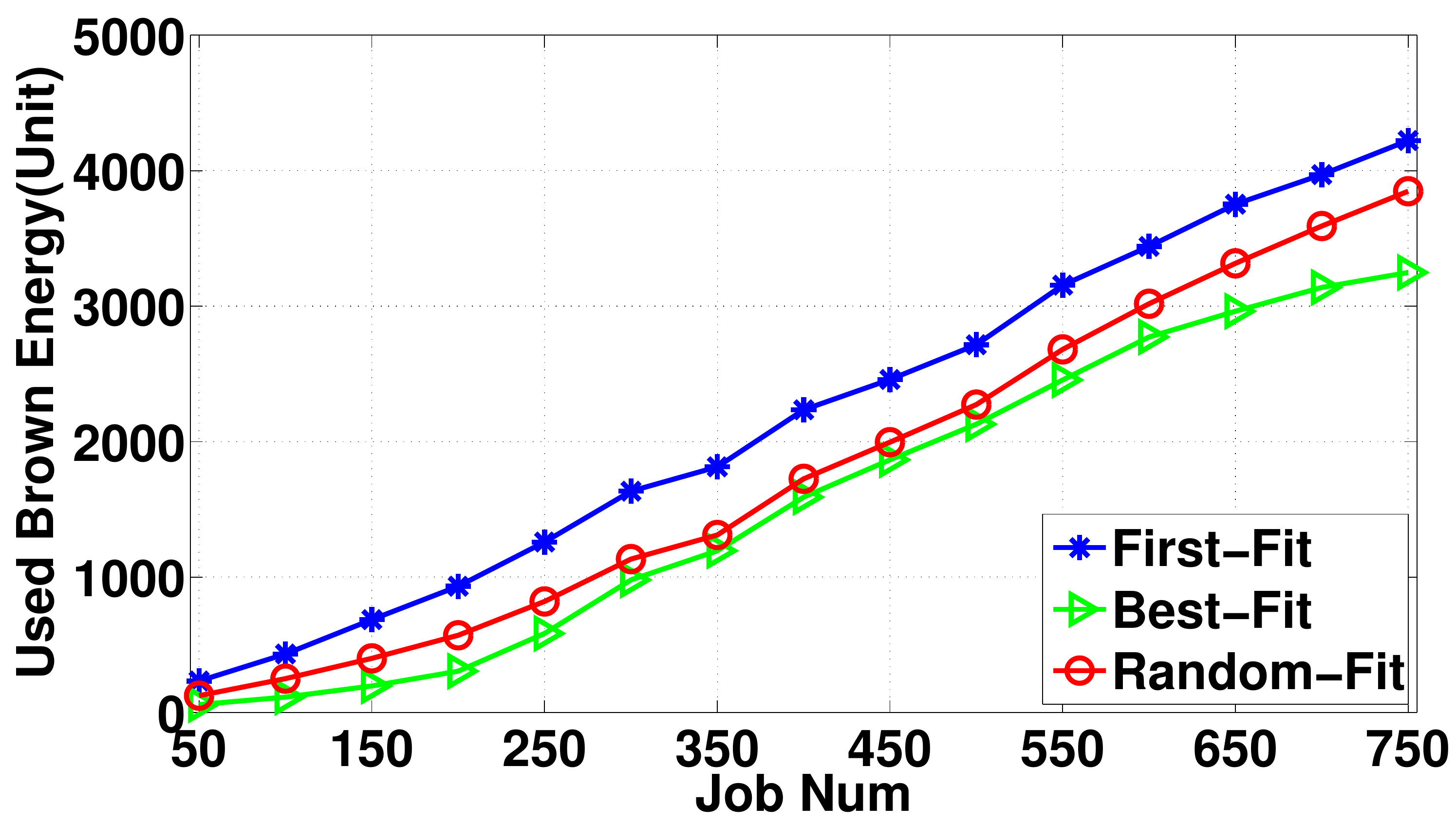}
\label{fig:UsedBrownEnergy_real}}
\caption{
Brown energy consumption under different workloads.
(a) UUTrace.
(b) PETrace.
(c) RealTrace.}
\label{fig:usedBrownEnergy}
\end{figure*}

From Figure~\ref{fig:profit}, we see Best-Fit tends to gain a better profit when the data center utilization is lower than $50\%$, while First-Fit is better when the data center utilization is higher (about $60\%$). In whatever data center utilization, Random-Fit always guarantees a better worst-case performance.

That Best-Fit is less profitable when the data utilization is high is because Best-Fit tends to delay scheduling jobs in order to consume less expensive energy. This delay scheduling behavior results in many jobs missing their deadlines and thus achieving a lower profit. While First-Fit always schedules jobs to the first available time slots thus it could schedule more jobs as that shown in Figure~\ref{fig:workload_scheduled}. But it cannot make a good use of green energy when the data center is of low utilization as that shown in Figure~\ref{fig:usedGreenEnergy}.

Taking the above analysis one step further, we conclude that if the data center utilization is predictable, then an adaptive scheduling algorithm which dynamically switches between Best-Fit and First-Fit according to the data center's utilization would have better performance than all the three algorithms. However, the data center utilization is usually hard to be predicted~\cite{MeisnerW10}.


\subsubsection{Comparisons with offline algorithms}

We also compare First-Fit, Best-Fit and Random-Fit against an optimal offline algorithm which is used as a benchmark with the goal of experimentally justifying the theoretical results.

As the optimal offline algorithm is computational hard, we do not include the performance of the offline algorithm in all various settings under which we compare the online algorithms. Instead, to have a brief understanding of the competitive ratio of  online algorithms, we drive the lower bound of competitive ratio based on their simulated performance.
First, we set the most profitable algorithm at each setting (under various workload utilizations) as an optimal performance $OPT'$. Then we compute the lower bound of competitive ratio using $OPT' / ALG$ where $ALG$ is the net profit gained by an online algorithm. As $OPT'$ is usually lower than the true optimal, therefore, the competitive ratio derived is only a lower bound of the real competitive ratio. It is fair enough to show that Random-Fit has a better worst-case competitive ratio than First-Fit and Best-Fit. The derived lower bound of competitive ratio from Figure~\ref{fig:profit_uniform_uniform}, \ref{fig:profit_uniform_equal} and Figure~\ref{fig:profit_poisson_equal} are shown in Figure~\ref{fig:cr}.

\begin{figure*}[h!]
\centering
\subfigure[]{\includegraphics[width=.32\textwidth]{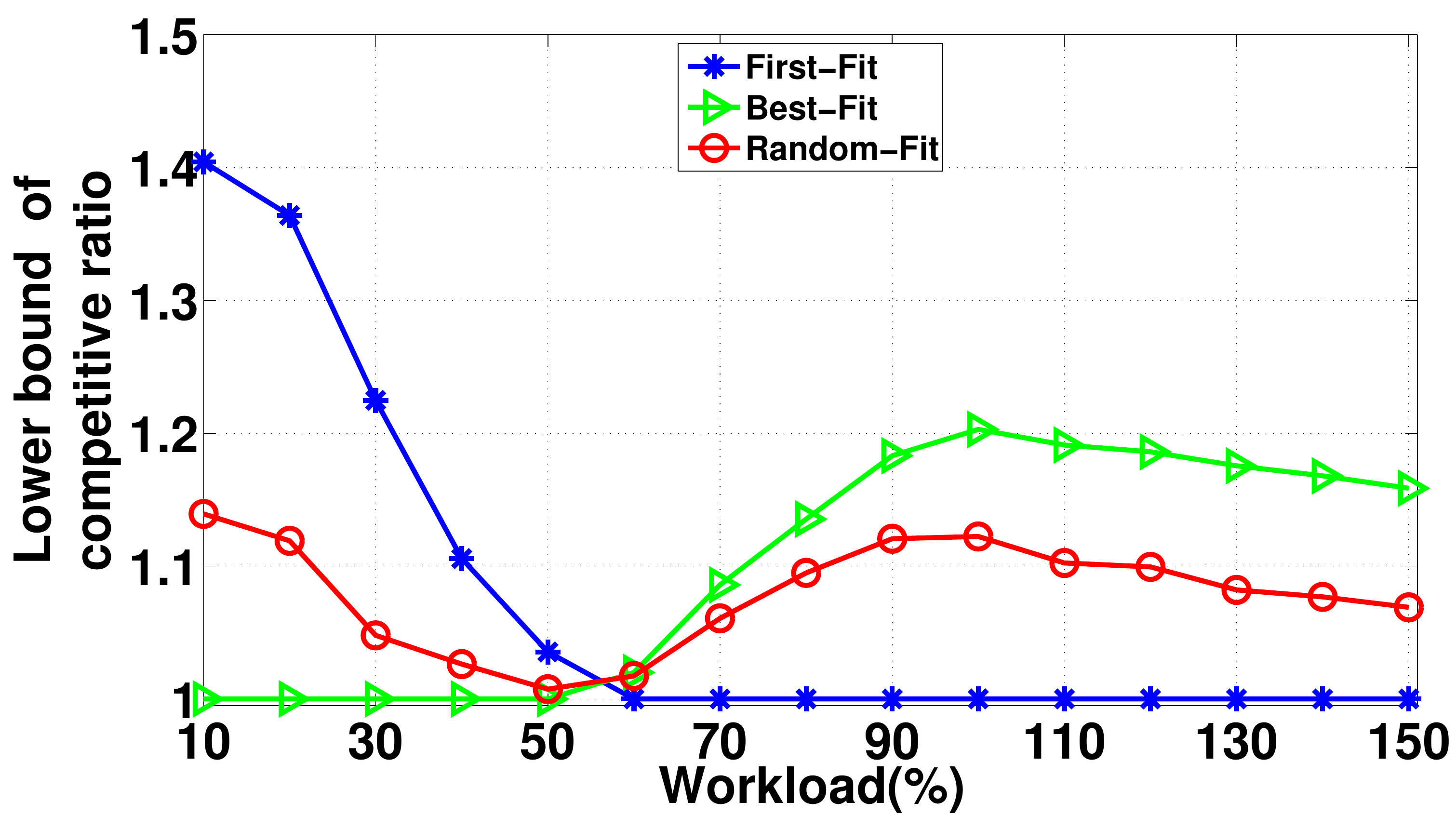}
\label{fig:cr_uniform_uniform}}
\subfigure[]{\includegraphics[width=.32\textwidth]{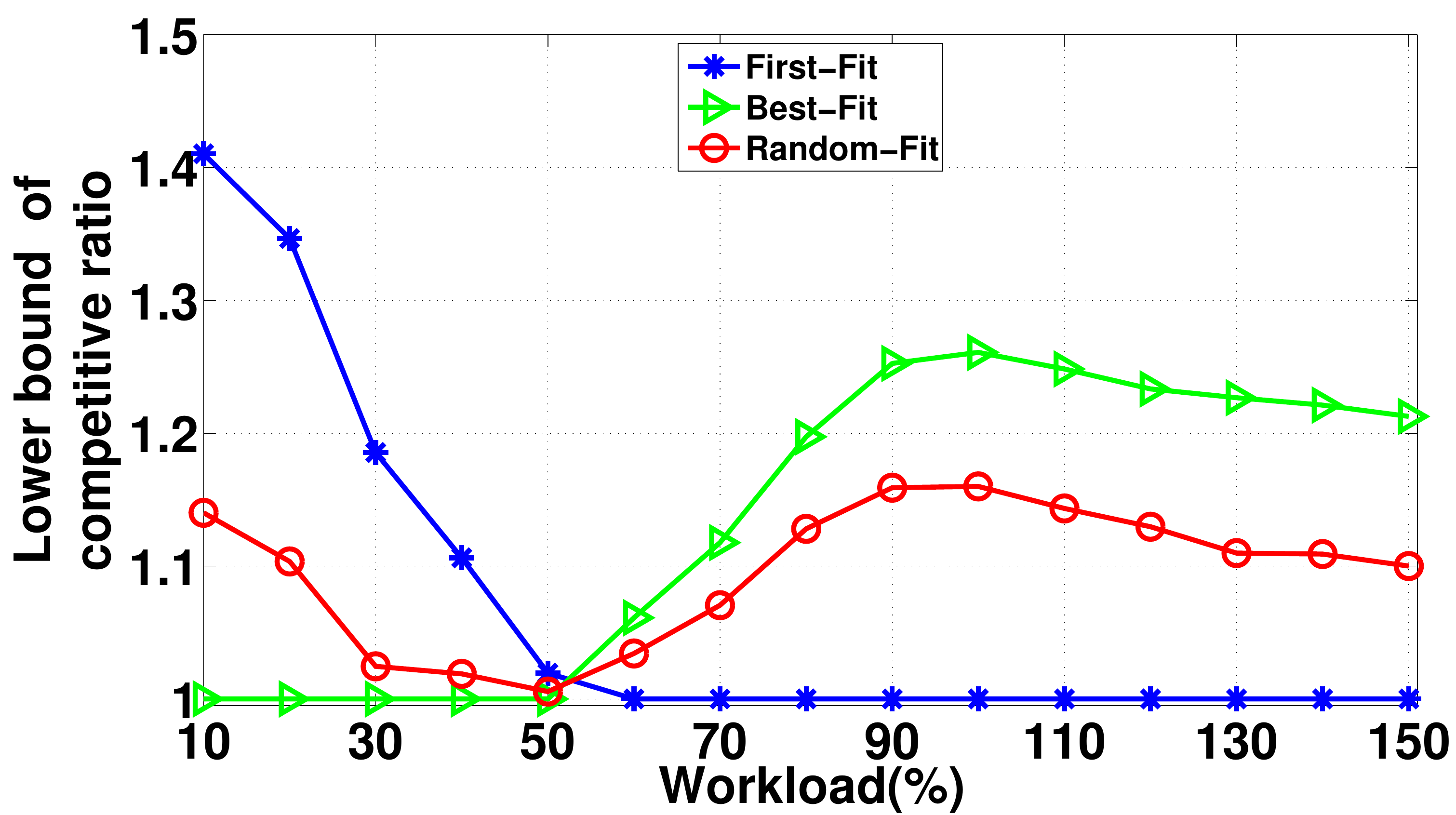}
\label{fig:cr_uniform_equal}}
\subfigure[]{\includegraphics[width=.32\textwidth]{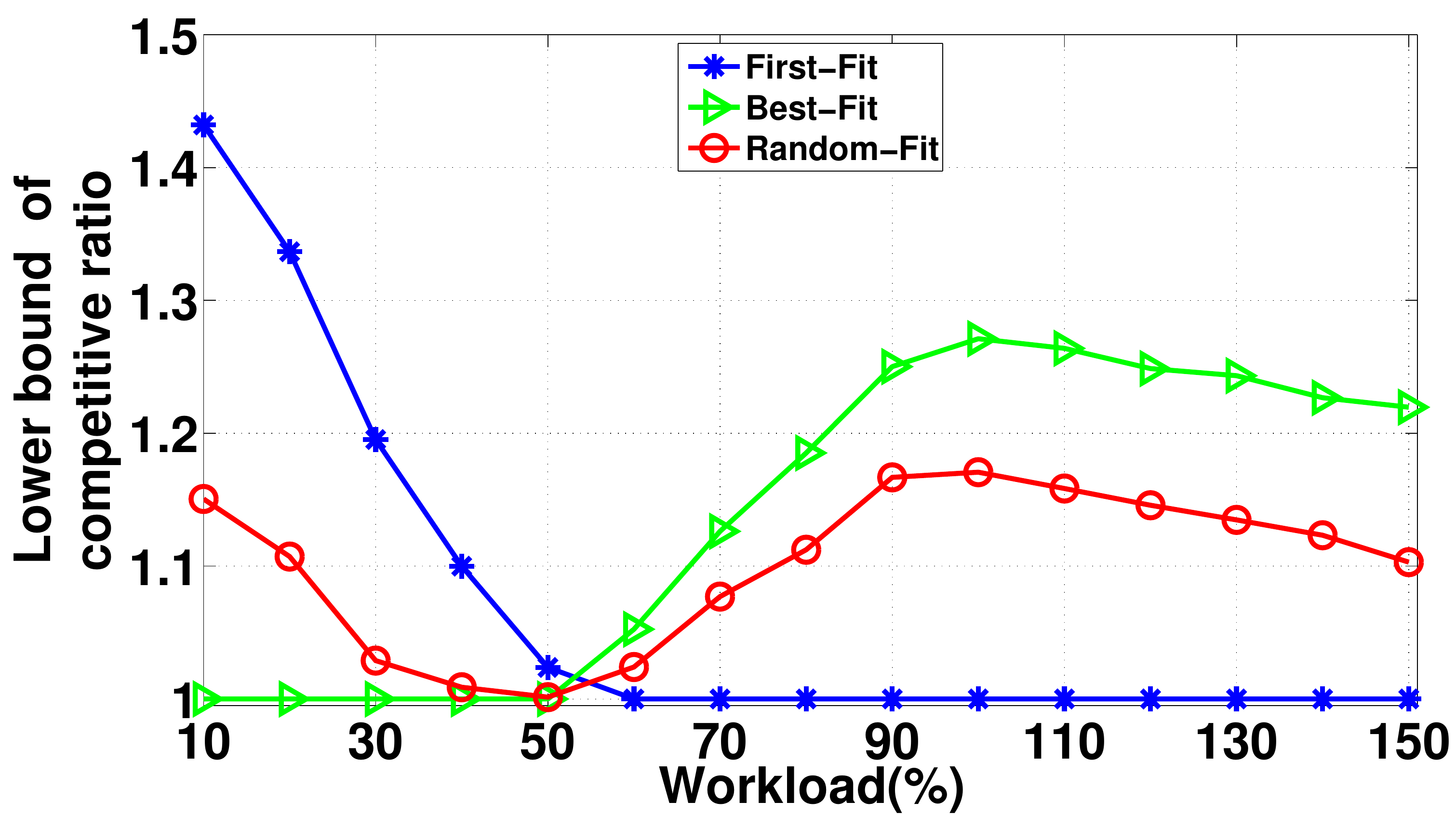}
\label{fig:cr_poisson_equal}}
\caption{
Lower bound of competitive ratio under different workloads.
(a) UUTrace.
(b) UETrace.
(c) PETrace.}
\label{fig:cr}
\end{figure*}

For \textbf{UETrace}, we implement an optimal offline algorithm to show the real experimental competitive ratios. The offline algorithm is described in the appendix and it is implemented using LINDO solver~\cite{lindo}. In the simulation, we simulate $2$ \textbf{UETrace} workloads with utilization $10\%$ and $100\%$ respectively. In both workloads, each job has a required processing time $5$ and a required node number $3$. We compare the online algorithms and the optimal offline algorithm on \textbf{the number of jobs scheduled}, \textbf{utilization of workloads scheduled}, \textbf{net profit earned}, \textbf{consumption of green/brown energy}, and \textbf{brown energy cost}. The results are shown in Table~\ref{tb_com_opt_20} and Table~\ref{tb_com_opt_100} respectively.

\begin{table}[ht!]
\centering
\begin{tabular}{|l|l|l|l|l|}  \hline
matrix & FF & BF & RF & OPT\\ \hline \hline
\# of scheduled jobs & 50.8 & 50.8 & 50.8 & 50.8\\ \hline
scheduled workload(\%) & 9.92 & 9.92 & 9.92 & 9.92\\ \hline
scheduled profit (\$) & 2.76 & 3.92 & 3.42 & 4.19\\ \hline
\# of green energy used & 311.2 & 678.6 &  524.4 & 762\\ \hline
\# of brown energy used & 450.8 & 83.4 & 237.6 & 3\\ \hline
brown energy cost (\$) & 1.43  & 0.27 & 0.77  & 0.0084\\ \hline
competitive ratio & 1.518  & 1.069 & 1.225 & 1\\ \hline
\end{tabular}
\caption{Comparison of online and offline algorithms using \textbf{UETrace} with loading factor $10\%$.}
\label{tb_com_opt_20}
\end{table}

\begin{table}[ht!]
\centering
\begin{tabular}{|l|l|l|l|l|}  \hline
matrix & FF & BF & RF & OPT\\ \hline \hline
\# of scheduled jobs  &  460 &  319.5 & 380.2 & 479 \\ \hline
scheduled workload(\%)  &  89.8 & 62.40  & 74.25 & 93.55 \\ \hline
scheduled profit (\$)  &  18.69  & 14.78  & 16.23 & 19.43  \\ \hline
\# of green energy used &  1402.8  & 1386.4 &  1398.1 & 1404 \\ \hline
\# of brown energy used & 5497.2 &  3406.1  & 4304.9 & 5781 \\ \hline
brown energy cost (\$) & 19.26  & 11.58 &  15.14  & 20.09 \\ \hline
competitive ratio & 1.04  & 1.31 &  1.20  & 1 \\ \hline
\end{tabular}
\caption{Comparison of online and offline algorithms using \textbf{UETrace} with loading factor $100\%$.}
\label{tb_com_opt_100}
\end{table}

From both Table~\ref{tb_com_opt_20} and Table~\ref{tb_com_opt_100}, we conclude that First-Fit and Best-Fit have competitive ratio worse (in this case, $1.518$ and $1.31$ respectively) than the theoretical upper bound ($1.25$) of Random-Fit. This conclusion confirms our theoretical result.


\subsubsection{Comparison of algorithm in preemptive and non-preemptive settings}

In the above subsections, we simulate the algorithms under the setting that job preemption is not allowed. In the following, we experimentally study the benefit of allowing jobs to be preempted. Allowing job preemption means that a job is not necessarily executed in a set of consecutive time slots. This side condition allows a scheduling policy to admit a better workload utilization and thus has the potential to gain more profit. In our setting, allowing job preemption may increase consumption of green energy.

In the following, we simulate preemption version of online algorithms, i.e., Preemptive First-Fit (PFF), Preemptive Best-Fit (PBF), and Preemptive Random-Fit (PBF) under $6$ types of workloads.

\begin{figure*}[h!]
\centering
\subfigure[]{\includegraphics[width=.32\textwidth]{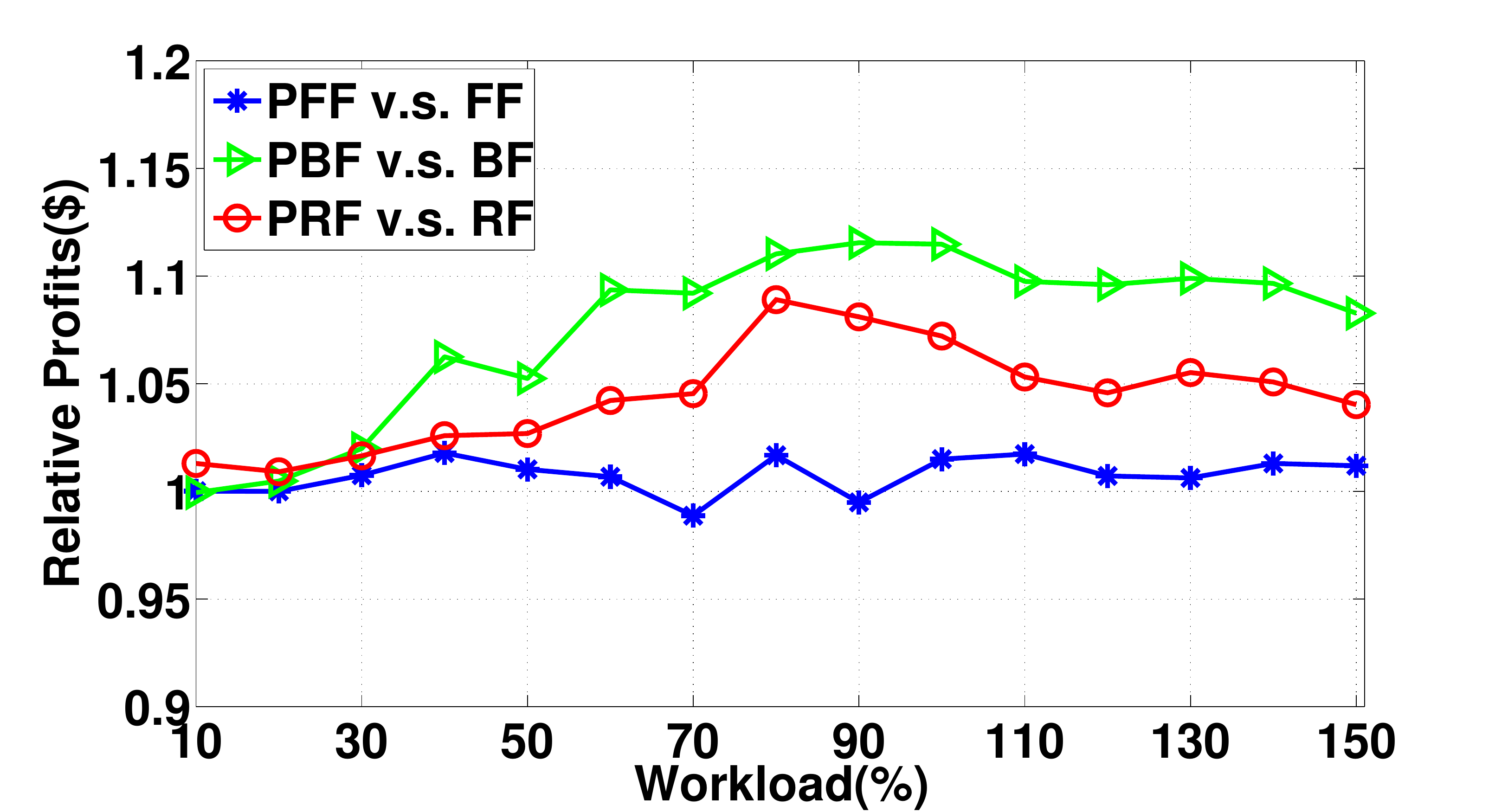}
\label{fig:Prempt_profit_uniform_uniofrm}}
\subfigure[]{\includegraphics[width=.32\textwidth]{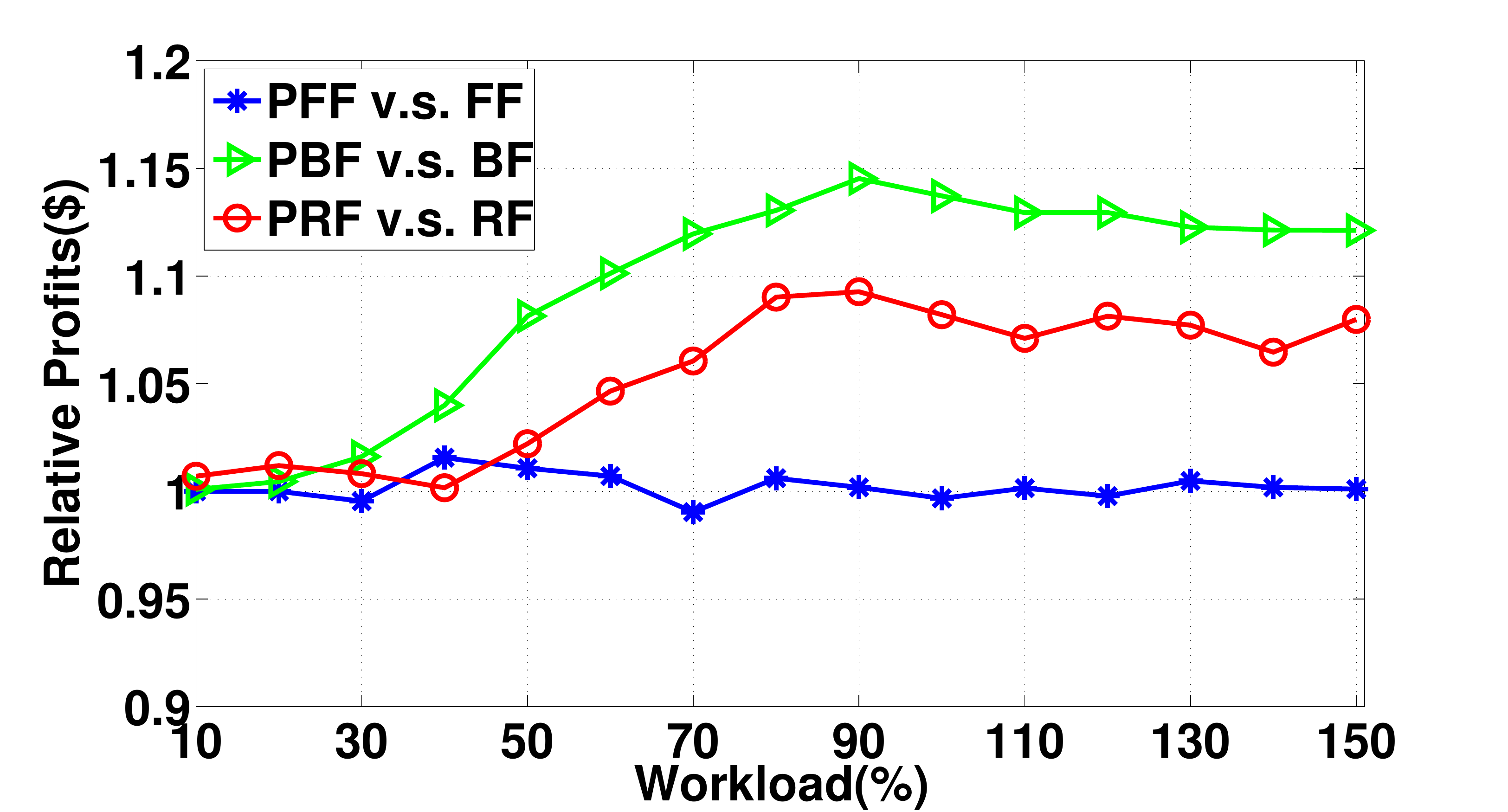}
\label{fig:Prempt_profit_uniform_equal}}
\subfigure[]{\includegraphics[width=.32\textwidth]{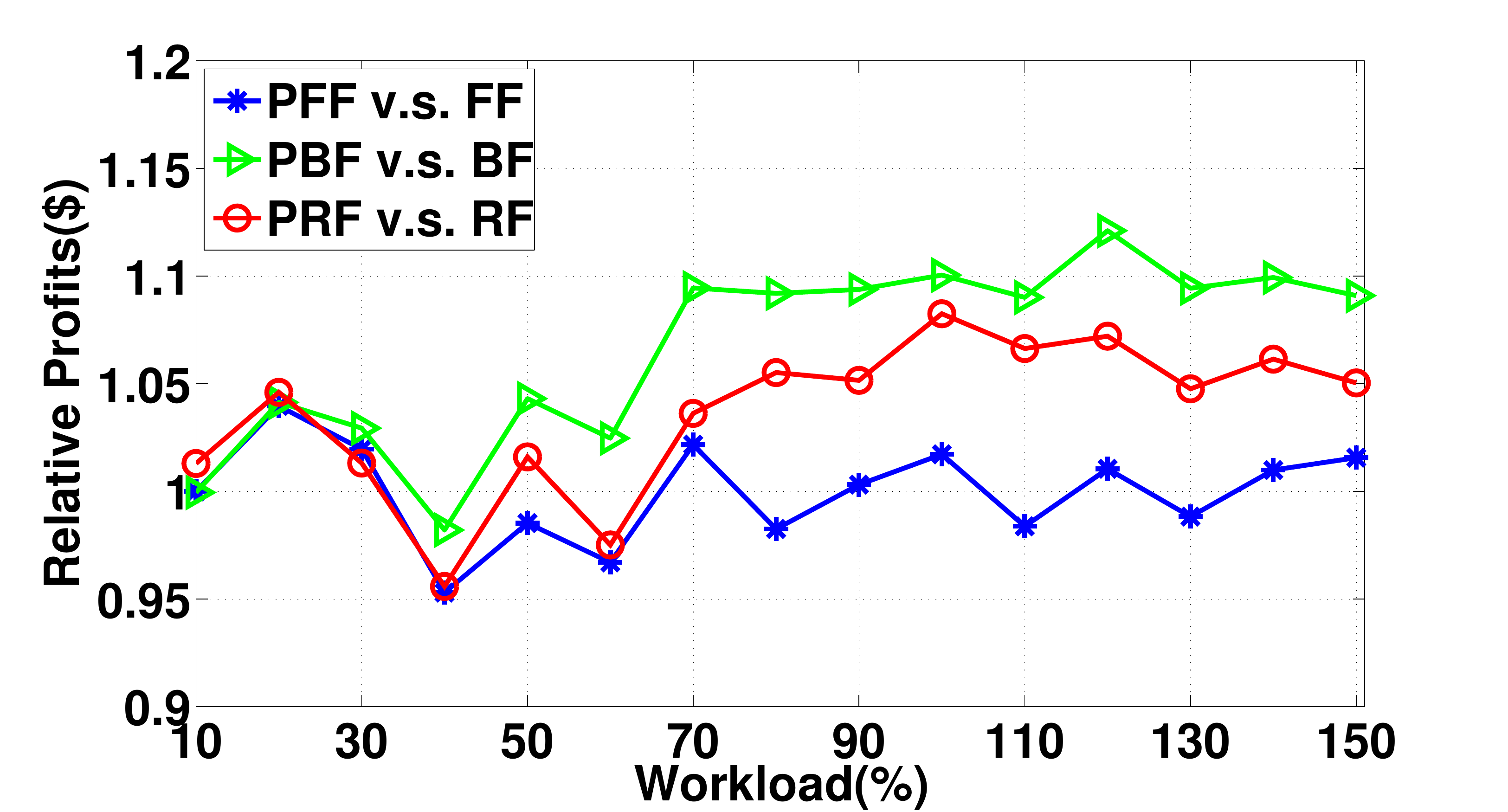}
\label{fig:Prempt_profit_poisson_uniform}}
\subfigure[]{\includegraphics[width=.32\textwidth]{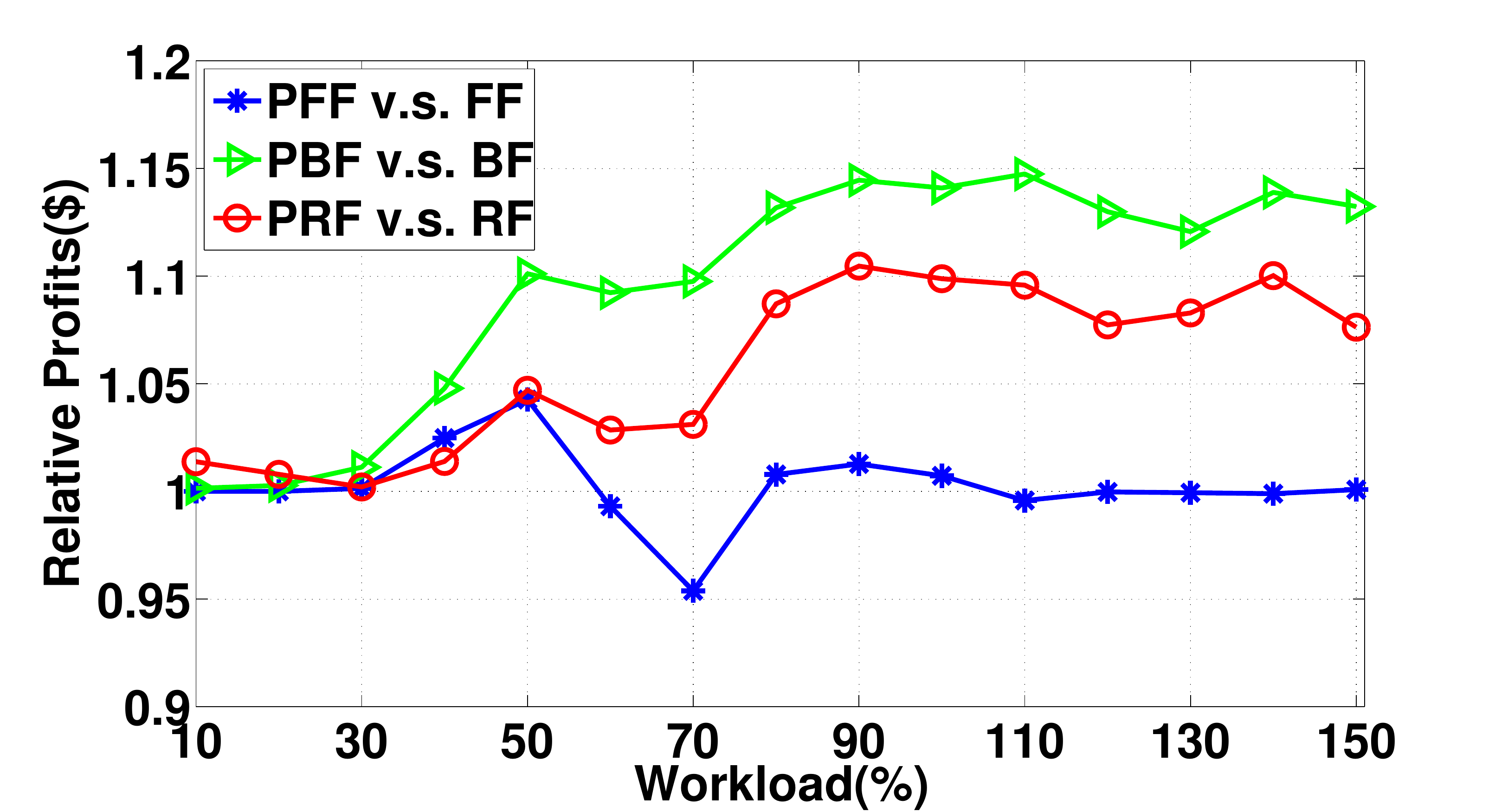}
\label{fig:Prempt_profit_poisson_equal}}
\subfigure[]{\includegraphics[width=.32\textwidth]{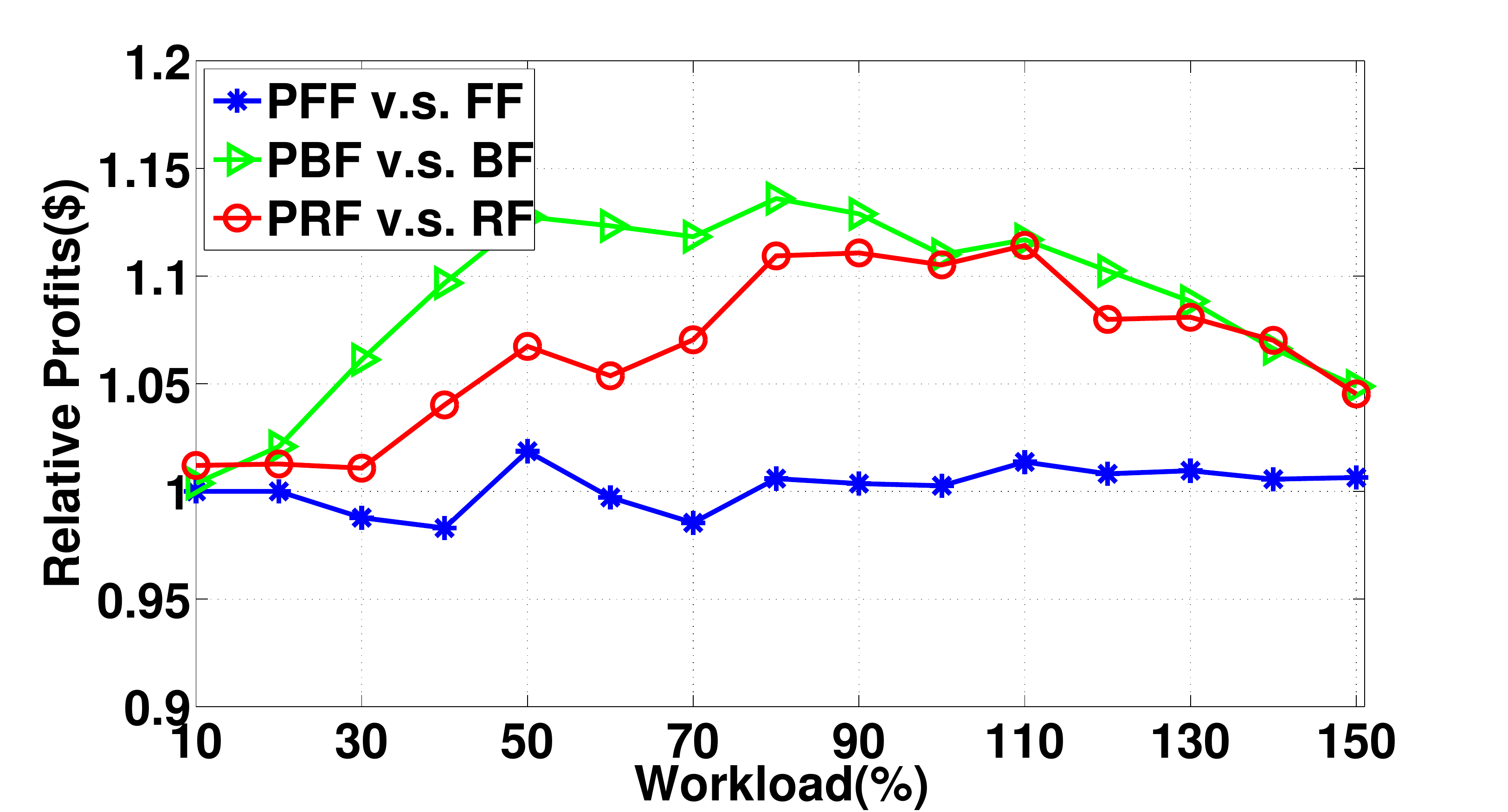}
\label{fig:Prempt_staggered_uniform}}
\subfigure[]{\includegraphics[width=.32\textwidth]{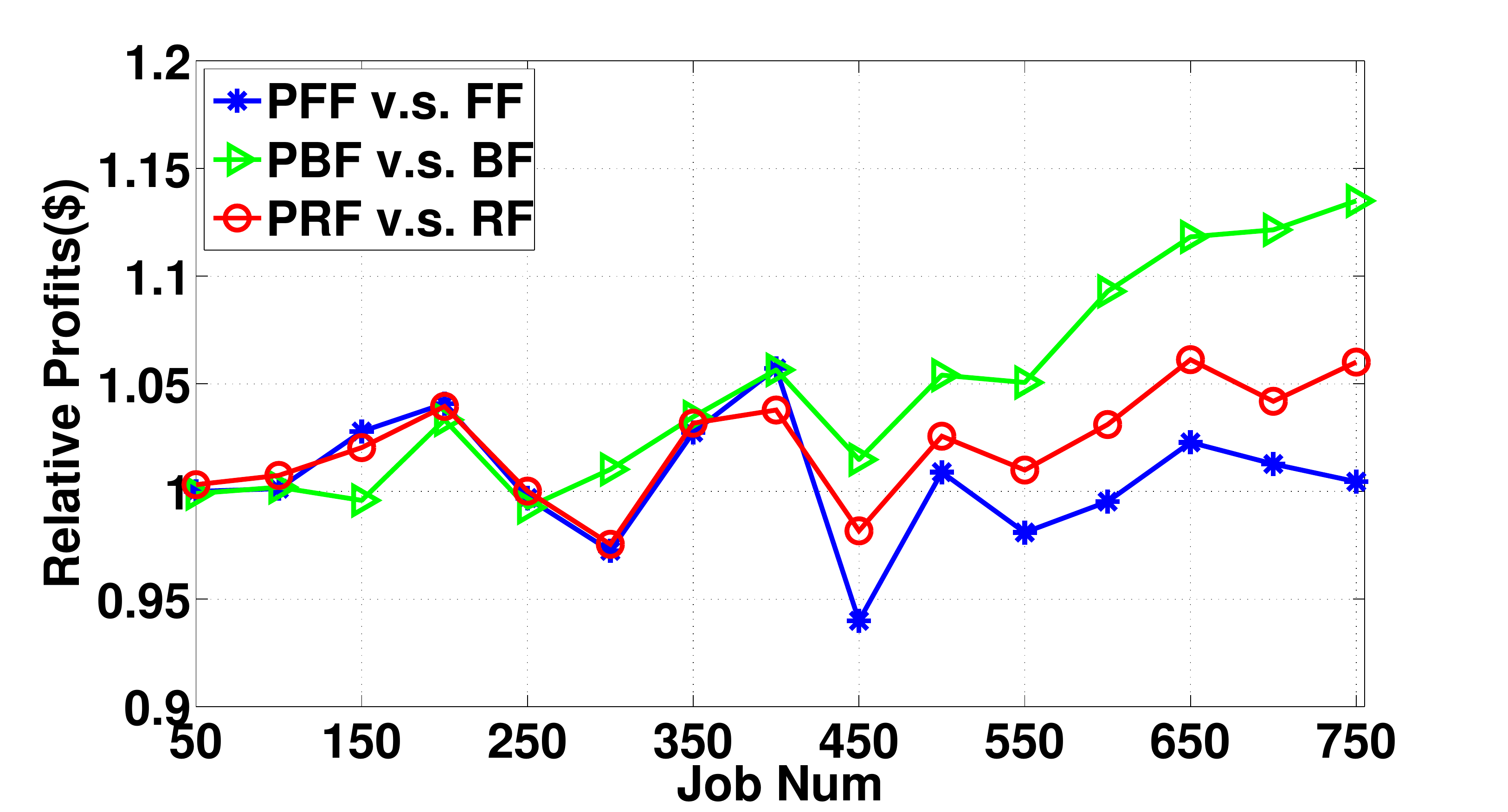}
\label{fig:Prempt_real_uniform}}
\caption{
Relative scheduled profit of preemptive algorithms (PFF, PBF, PRF) and their corresponding non-preemptive algorithms (FF, BF, RF)  under various workload traces.
(a) UUTrace.
(b) UETrace.
(c) PUTrace.
(d) PETrace.
(e) StaggeredTrace.
(f) RealTrace.}
\label{fig:Prempt_profit}
\end{figure*}

From Figure~\ref{fig:Prempt_profit}, we see that allowing job preemption has different impacts on the performance of  algorithms. And the impact also depends on the data center's utilization. For First-Fit, job preemption makes the profit worse. The reason underlying this observation is that job preemption happens at night time usually incurs more brown energy cost at night. For Best-Fit, job preemption helps improve profit when the data center's utilization is high. The reason is because that when the utilization is high, Best-Fit is more likely to miss jobs' deadlines if no job preemption is allowed. For Random-Fit, in most case, job preemption helps when the data center's utilization is high. The reason is the same as that for the Best-Fit algorithm.

In summary, whether job preemption helps or not depends on both the  algorithm used and the data center's utilization. Therefore, at the time when making decisions on whether to make job preempted or not, it would be better to study the nature of the algorithms and to take the data center's utilization into account.


\section{Conclusions}

In this paper we study online scheduling of energy and jobs on multiple machines in a green data center with the objective of maximizing net profit of service providers. This decision-making problem involves three questions: (1) whether to admit a job, (2) when to schedule this job, and (3) which machines and which type of energy designated to run it. In our problem setting, costs are time-sensitive and so is the net profit. Previous work employs deterministic approaches only and the underlying algorithmic ideas are either First-Fit or Best-Fit; and no theoretical analysis has been given. In this work, competitive analysis is used to measure an online algorithm's theoretical performance. An algorithm with a better competitive ratio has better worst-case performance.  We conclude that randomness plays an important role in maximizing net profit in this setting. We adjust the probability and design a theoretically-better online algorithm.  Furthermore, we conduct experiments on both real and simulated workload traces to show that our algorithm indeed outperforms the previous ones, as what the theory indicates.

In our future work, we will extend the randomness idea in the most general setting and the geographically distributed data centers. We will also study delay-sensitive revenue management for green data centers. As what we ever discussed in Section~\ref{subsubsec_simOnline}, a mixture of algorithms under various utilization may achieve better net profit. We will also investigate this problem.


\section{Acknowledgments}

This material is based upon work supported by NSF under Grants No. CCF-0915681 and CCF-1216993. Any opinions, findings, and conclusions or recommendations expressed in this material are those of the authors and do not necessarily reflect the views of NSF.


\bibliographystyle{plain}
\bibliography{../greenSlot}

\appendix

\section{Hardness of the Problem GDC-RM}

Note that the revenue management problem essentially is not an offline problem since the jobs and the green energy cannot be modeled and predicted precisely at all the time. However, understanding the hardness of the offline problem may be useful to us in evaluating an online algorithm's theoretical and empirical performance. We prove that the offline version (with and without job preemption) is NP-hard. Using a reduction from the well-known NP-hard problem `Knapsack'~\cite{GareyJ79}, we conclude the following result.

\begin{theorem}
The problem GDC-RM with or without job preemption is NP-hard.
\label{thm:nphard}
\end{theorem}

\begin{proof}
Given a candidate solution, it takes polynomial-time for us to verify whether this solution is feasibly scheduled or not. Thus, the problem GDC-RM belongs to NP. In the following, we prove that GDC-RM is NP-hard by showing a polynomial-time reduction from the Knapsack problem to it. In the Knapsack problem, there are a knapsack of capacity $W$ and $n$ items with each one has size $s_i$. The goal is to make the knapsack as full as possible. The Knapsack problem is known NP-hard~\cite{GareyJ79}.

Consider the problem GDC-RM. Assume the produced green energy has a budget of $B$ in a scheduling window and the brown energy's costs ($B^d$ and $B^n$) are high enough such that any use of brown energy makes no positive net profit at all. Therefore, to maximize the net profit, we would like to find a set of jobs such that these jobs consume as much as close to but no more than the green energy budget $B$ without using any amount of the brown energy. Particularly, we restrict that the green energy is available within a scheduling window $[t, t']$ and all jobs have the same release time $t$ and the same deadline $t'$, which are the boundaries of this scheduling window. Let $t' - t = W$. Also, we restrict that each job $j$ has $q_j = 1$. This conversion takes linear time of the number of jobs.

If we have a polynomial-time optimal solution to the problem GDC-RM with the special input instance as created as in the above, then we have an optimal solution to the following Knapsack problem: The knapsack has its capacity of $W = t' - t$ and each item $j$ has its size of $p_j$. As the Knapsack problem is NP-hard, then the problem GDC-RM with or without job preemption is NP-hard. Recall that all the jobs have the same release time and the same deadline, thus, any preempted schedule can be converted into a non-preempted schedule with running time $O(n \log n)$, where $n$ is the number of jobs. So, this conversion fits for both job preemptive and job non-preemptive settings.
\end{proof}

We further strengthen one result described in Theorem~\ref{thm:nphard} and show the following result about the job non-preemptive setting.

\begin{corollary}
The problem GDC-RM without job preemption is strongly NP-hard.
\end{corollary}

\begin{proof}
The offline problem of GDC-RM without job preemption can be reduced from the problem 3-Partition. In the problem 3-Partition, we have a set of numbers and our objective is to find 3 subsets of numbers such that their total values are the same. The 3-Partition problem is known to be strongly NP-hard such that unless P $=$ NP, there does not exist a polynomial-time approximation scheme~\cite{CormenLRS09}.

Consider an instance of the problem 3-Partition with $n$ numbers $\{s_1, s_2, \ldots, s_n\}$. Each subset has its total value of $\sum_u s_i / 3$. The reduction is constructed as below. For an instance of the problem GDC-RM, we still have a green energy budget and the brown energy's cost is astronomical such that we won't use any brown energy. There are two special jobs $j$ and $j'$ such that their release times plus their processing times equal to their deadlines: $r_j + p_j = d_j$ and $r_{j'} + p_{j'} = d_{j'}$. For the remaining jobs, they share the same release time and the same deadline. The processing time of the $i$-th ($i \neq j, j'$) job is denoted as $s_i$. Their common release time is $0$ and their common deadline is $T = \sum_i s_i + p_j + p_{j'}$. As we see that in this instance, the jobs $j$ and $j'$ must be scheduled immediately when they are released at time $r_j$ and $r_{j'}$. For the two special jobs $j$ and $j'$, we set their release times $\sum^n_{i = 1} s_i / 3$ and $\left(2 / 3\right)\sum^n_{i = 1} s_i + p_j$ respectively. These two jobs $j$ and $j'$ partition the time interval $[0, T]$ into 5 sub-intervals $T_1 := [0, r_j)$, $T_2 := [r_j + p_j)$, $T_3 := [r_j + p_j, r_{j'})$, $T_4 := [r_{j'} + p_{j'})$, and $T_5 := [r_{j'} + p_{j'}, T]$ with sub-intervals $T_1$, $T_3$ and $T_5$ having their sizes equal to $|T_1| = |T_3| = |T_5| = \sum_i s_i / 3$. The jobs rather than $j$ and $j'$ should be scheduled within $T_1$, $T_3$ and $T_5$.

Note that if we have an optimal solution to the special instance of the problem GDC-RM without job preemption in polynomial-time, then we have a optimal solution to the problem of 3-Partition. Thus, the problem of GDC-RM without job preemption is strongly NP-hard.
\end{proof}


\section{Competitive Analysis of First-Fit and Best-Fit}

According to the definition of profit, a job $j$ with $p_j$ processing time and $q_j$ node requirement has profit
\begin{displaymath}
c \cdot p_j \cdot q_j - \int_t P(t),
\end{displaymath}
where $P(t)$ has the value $0$ (for green energy), $B^d$ (for on-peak brown energy), or $B^n$ (for off-peak brown energy) respectively when the job is processed using various types of energy. $P(t)$ is in integral along the time when the machines process $j$. If all jobs are with the same processing time and node requirements, then we normalize the profit as
\begin{displaymath}
\frac{c \cdot p_j \cdot q_j - \int_t P(t)}{c \cdot p_j \cdot q_j} = 1 - \int_t \frac{P(t)}{c \cdot p_j \cdot q_j}.
\end{displaymath}

In our proofs below, we generate instances such that for each job, it is processed by \emph{only} one type of energy using the particular algorithm. Thus, for ease to present the competitive ratio, we define $1 - \frac{P(t)}{c \cdot p_j \cdot q_j}$ as $v_{on}$, $v_{off}$, $v_g$ as below.

\begin{displaymath}
1 - \frac{P(t)}{c \cdot p_j \cdot q_j} :=
\begin{cases}
v_{on}, & \text{if only using on-peak brown energy to schedule $j$}\\
v_{off}, & \text{if only using off-peak brown energy to schedule $j$}\\
v_g, & \text{if only using green energy to schedule $j$}
\end{cases}
\end{displaymath}

Note that the normalized profit $1 - \frac{P(t)}{c \cdot p_j \cdot q_j}$ has a value among $(0, 1]$. According to the fact that on-peak brown energy is expensive than off-peak brown energy. Also, green energy has cost $0$. We have $0 < v_{on} < v_{off} < v_{g} = 1$. Also, for jobs with the same processing times and same node requirements, they have the same value for $v_{on}$, $v_{off}$, and $v_g$.

\begin{theorem}
The lower bound of competitive ratio for First-Fit is $\max \left\lbrace\frac{v_{off}}{v_{on}}, \frac{v_{g}}{v_{on}}\right\rbrace$.
\label{comRatio_FF}
\end{theorem}

\begin{proof}
To prove the lower bound, we create an input instance. Let OPT denote an optimal offline algorithm. We assume each job has processing time requirement $p_j = 1$ and node requirement $q_j = M$. We use $(r, d)$ to denote a job with release time $r$ and deadline $d$. We have $M$ machines.

Assume there are two daytime time slots $t_1$ and $t_2$, with $0$ and $M$ green energy units arriving at them respectively. Assume there is only one job $j = (t_1, t_2)$ arriving. First-Fit schedules $j$ at time $t_1$, earning a revenue $v_{on}$. OPT schedules $j$ at time $t_2$, achieving a profit $v_g$. The competitive ratio is $\frac{OPT}{FF} = \frac{v_{g}}{v_{on}}$.

If $t_1$ is at on-peak and $t_2$ is at off-peak, then we assume that no green energy arrives at both time slots. Using the same analysis approach, we get the competitive ratio $\frac{OPT}{FF} = \frac{v_{off}}{v_{on}}$. Therefore, we conclude that First-Fit has a competitive ratio at least $\max\left\lbrace \frac{v_{off}}{v_{on}}, \frac{v_{g}}{v_{on}}\right\rbrace$.
\end{proof}

\begin{theorem}
The lower bound of competitive ratio for Best-Fit is $\max\left\lbrace 1+\frac{v_{on}}{v_{off}}, 1 + \frac{v_{off}}{v_{g}}\right\rbrace$.
\label{comRatio_BF}
\end{theorem}

\begin{proof}
We prove via constructing an input instance as a lower bound example. We assume all the arriving jobs have processing time $p_j = 1$ and node requirement $M$ (no two jobs can be executed simultaneously at the same time slot).

Assume there are two time slots $t_1$ and $t_2$ --- $t_1$ is at on-peak while $t_2$ is at off-peak. There are no green energy arriving at both time slots. Assume there are two jobs released $j_1 = (t_1, t_2)$ and $j_2 = (t_2, t_2)$.

Best-Fit will delay job $j_1$ to be scheduled at time $t_2$, resulting in a deadline conflict between jobs $j_1$ and $j_2$, and thus only gain profit $v_{off}$. OPT will schedule $j_1$ and $j_2$ at time $t_1$ and $t_2$ respectively, gaining a profit $v_{on} + v_{off}$. Thus the competitive ratio is $\frac{OPT}{BF} = 1 + \frac{v_{on}}{v_{off}}$.

If $t_1$ is at off-peak and $t_2$ is at on-peak, then we assume there are $0$ and $M$ units of green energy arrive at time $t_1$ and $t_2$ respectively. Using the same analysis approach, we get the competitive ratio $\frac{OPT}{BF} = 1 + \frac{v_{off}}{v_{g}}$. we conclude that Best-Fit has a competitive ratio at least $\max\left\lbrace 1 + \frac{v_{on}}{v_{off}}, 1 + \frac{v_{off}}{v_{g}}\right\rbrace$.
\end{proof}

Based on the above analysis and recall $0 < v_{on} < v_{off} < v_{g} = 1$, we have the following result.

\begin{corollary}
Deterministic algorithms First-Fit and Best-Fit, with or without job preemption, have competitive ratios no strictly better than $2$, even for a restricted case in which all jobs are with the same length.
\end{corollary}


\section{Optimal Offline Algorithms}

We consider optimal offline algorithms in this section. Let $J$ denote the job set in an instance.


\subsection{Job preemptive setting}

We consider the setting in which jobs are preemptive. We formulate the optimization problem as NLIP (non-linear integer program). Let $y_j$ be the indicator variable about whether job $j$ is selected or not --- $y_j = 1$ means job $j$ is selected. Let $x_{jmt}$ be the indicator variable about whether job $j$ is being executed at node $m$, time $k$, $x_{jmt} = 1$ means job $j$ is being executed at node $m$ at time $t$. Let $b(t)$ denote the unit brown energy price at time $t$. The total brown energy cost is denoted as $\emph{EnergyCost}$. The optimization problem is to optimization profit while satisfying the following constraints.
\begin{align}
\max & R - E &\\
\mbox{subject to} & R = \sum_{j \in J} v_j \cdot y_j &\\
& E = \sum^T_{t = 1} \max\{0, \sum_{m, j} x_{jmt} - g(t)\} \cdot b(t) &\\
& \sum_{j \in J} x_{j m t} \le M & \forall m, t\\
& \sum^T_t x_{jmt} = \{p_j \mbox{ or } 0\} \cdot y_j & \forall j, m\\
& \sum^M_m x_{jmt} = \{q_j \mbox{ or }  0\}\cdot y_j & \forall j, t\\
& \sum^T_{t = 1} \sum_r x_{jmt} = p_j \cdot q_j \cdot y_j & \forall j\\
& \sum^T_{t > d_j} x_{jmt} = 0 & \forall j, m\\
& \sum^T_{t < r_j} x_{jmt} = 0 & \forall j, m\\
& x_{jmt} = \{0, 1\} & \forall j, m, t\\
& y_j = \{0, 1\} & \forall j, m, t
\end{align}

We briefly explain the meaning of each constraint function in the following: constraint $(4)$ indicates the number of nodes used at any time should be less than the capacity of nodes available; constraint $(5)$ means if a job is scheduled, then its active time at any node should be either its required processing time $p_j$ or $0$, in other words, partial execution at one node is not allowed (migration between nodes is not allowed); constraint $(6)$ means if a job is scheduled, then at any time, it should be active on $q_j$ nodes or $0$ where $q_j$ is its required node numbers during execution, i.e., this constraint guarantee the parallel execution of a job; constraint $(7)$ guarantees that the scheduled job should be finished; constraints $(8)$ and $(9)$ are the arrival time and deadline constraints, i.e., no job is allowed to execute before its release time or after its deadline.


\subsection{Job non-preemptive setting with same job processing times and node requirements}

We formulate a linear program for the special cases when jobs have same processing times and node requirements. Let $g(t)$ denotes the amount of green energy arrive at time $t$ and let $b(t)$ denotes the unit brown energy price at time $t$. Assume all jobs have the same processing time slots $p$ and node requirement $q$. Let $y_j$ be an indicator variable indicates whether a job is scheduled ($y_j = 1$) or not ($y_j = 0$). Let $s[j, t]$ be an indicator variable denotes whether job $j$ is started at time $t$ ($s[j, t]=1$) or not ($s[j, t] = 0$). Let $n(t)$ denotes the jobs started at time $t$. Let $e(t)$ denotes the energy demand at time $t$. Then we have
\begin{align*}
\max & R - E &\\
\mbox{subject to} & R = \sum_{j \in J} v_j \cdot y_j &\\
& E = \sum^T_{t = 1} \max\{0, e(t) - g(t)\} \cdot b(t) &\\
& n(t) = \sum^J_{j = 1} s[j, t] & \forall j\\
& e(t) = \sum_{\max\{0, n - p + 1\} \le k \le t} n(k) \cdot q & \forall t\\
& e(t) \le M & \forall t\\
& \sum^T_{t = 1} s[j, t] \ge y_j & \forall j\\
& \sum^T_{t > d_j} s[j, t] = 0 & \forall j\\
& \sum^T_{t < r_j} s[j, t] = 0 & \forall j\\
& s[j, t] = \{0, 1\} & \forall j, t\\
& y_j = \{0, 1\} & \forall j
\end{align*}


\end{document}